\documentclass[10pt,reqno]{amsart}

\usepackage{hyperref}
\usepackage{latexsym,amsmath, bm}
\usepackage{enumerate}
\usepackage{amsfonts}
\usepackage{amssymb} 
\usepackage{geometry}
\usepackage{latexsym}
\usepackage{fixmath}
\usepackage{faktor}
\usepackage{mathtools}

\usepackage[T1]{fontenc}
\usepackage{fourier}
\usepackage{bbm}
\usepackage{color}

\usepackage{pgfplots}
\usepackage{tikz}
\usetikzlibrary{shapes,decorations,arrows,calc,arrows.meta,fit,positioning}
\usepgfplotslibrary{fillbetween}
\usepackage{graphicx}
\usepackage{fullpage}
\usetikzlibrary{decorations.pathreplacing, matrix}

\usepackage[boxed, linesnumbered, noend, noline]{algorithm2e}
\SetKwInput{KwData}{Input}
\SetKwInput{KwResult}{Output}

\numberwithin{equation}{section}

\newcommand{\wxstar}[1]{\vW^\star_{{#1}}}

\newcommand{\wxj}[2]{\vW_{{#1},{#2}}}

\newcommand{\wq}{\sum_{j=1}^{s-1} 2^{j/s-1} w_j}

\renewcommand\log{\ln}

\newcommand\SPIV{{\tt SPIV}}
\newcommand\DD{{\tt DD}}

\newcommand{\vX}{\vec X}

\renewcommand{\epsilon}{\eps}

\newcommand\nix{\,\cdot\,}

\newcommand\vW{\vec W}
\newcommand\vS{\vec S}

\newcommand\dd{{\mathrm d}}

\renewcommand{\vec}[1]{\boldsymbol{#1}}

\newcommand\KL[2]{D_{\mathrm{KL}}\bc{{{#1}\|{#2}}}}

\newcommand\SIGMA{\vec\sigma}
\newcommand\CHI{\vec\chi}
\newcommand\TAU{\vec\tau}

\newtheorem{definition}{Definition}[section]
\newtheorem{claim}[definition]{Claim}

\newtheorem{theorem}[definition]{Theorem}
\newtheorem{lemma}[definition]{Lemma}
\newtheorem{proposition}[definition]{Proposition}
\newtheorem{corollary}[definition]{Corollary}

\newtheorem{fact}[definition]{Fact}

\newcommand\fE{\mathfrak{E}}

\newcommand\cA{\mathcal{A}}

\newcommand\cD{\mathcal{D}}

\newcommand\cE{\mathcal{E}}
\newcommand\cU{\mathcal{U}}
\newcommand\cN{\mathcal{N}}

\newcommand\cS{\mathcal{S}}

\newcommand\cL{\mathcal{L}}
\newcommand\cM{\mathcal{M}}

\newcommand\cP{\mathcal{P}}
\newcommand\cX{\mathcal{X}}

\newcommand\cW{\mathcal{W}}
\newcommand\cZ{\mathcal{Z}}

\def\cE{{\mathcal E}}

\newcommand\eul{\mathrm{e}}
\newcommand\eps{\varepsilon}

\newcommand\Erw{\mathbb{E}}
\newcommand{\vecone}{\vec{1}}

\newcommand{\Po}{{\rm Po}}
\newcommand{\Bin}{{\rm Bin}}
\newcommand{\Hyp}{{\rm Hyp}}

\newcommand\bc[1]{\left({#1}\right)}
\newcommand\cbc[1]{\left\{{#1}\right\}}
\newcommand\bcfr[2]{\bc{\frac{#1}{#2}}}

\newcommand\brk[1]{\left\lbrack{#1}\right\rbrack}

\newcommand\abs[1]{\left|{#1}\right|}

\newcommand{\Whp}{W.h.p.}
\newcommand{\whp}{w.h.p.}

\newcommand{\Erdos}{Erd\H{o}s}
\newcommand{\Renyi}{R\'enyi}

\newcommand{\Mezard}{M\'ezard}

\newcommand\pr{\mathbb{P}} 
\renewcommand\Pr{\pr} 

\newcommand\Lem{Lemma}
\newcommand\Prop{Proposition}
\newcommand\Thm{Theorem}

\newcommand\Cor{Corollary}
\newcommand\Sec{Section}

\newcommand{\ceil}[1]{\left\lceil#1\right\rceil}

 \def\G{{\vec G}}

\def\pr{{\mathbb P}}

\newcommand{\remove}[1]{}

\newcommand{\one}{V_1}

\newcommand{\zeroplus}{V_{0+}}
\newcommand{\oneplus}{V_{1+}}

\newcommand{\zeroplusi}[1]{V_{0+}[{#1}]}

\newcommand\madapt{m_{\mathrm{ad}}}
\newcommand\minf{m_{\mathrm{inf}}}
\newcommand\malg{m_{\mathtt{DD}}}

\newcommand\mDDB{m_{\mathtt{DD},\mathrm{Be}}}

\newcommand\mBL{m_{\mathrm{BL,ad}}}
\newcommand\mMT{m_{\mathrm{MTT,ad}}}
\newcommand\mQGT{m_{\mathrm{QGT}}}

\newcommand{\be}{\begin{equation}}
	\newcommand{\bel}[1]{\begin{equation}\lab{#1}\ }
		\newcommand{\ee}{\end{equation}}
	\newcommand{\bea}{\begin{eqnarray}}
		\newcommand{\eea}{\end{eqnarray}}
	\newcommand{\bean}{\begin{eqnarray*}}
		\newcommand{\eean}{\end{eqnarray*}}

\newcommand\FigG{
\begin{figure}
\begin{minipage}[t][][b]{0.95 \textwidth}
\begin{tikzpicture}[scale=0.8]
\begin{axis}[
axis lines = left,
 xlabel = density parameter $\theta$,
ylabel style={at={(0,0.3)}},
xtick={0, 0.409, 0.5, 1},
xticklabels={$0$, $\frac{\log2}{1+\log2}$, $\frac{1}{2}$, $1$},
ytick={2.081368981, 1.442695041, 1.040684490, 0.8520789316},
yticklabels={$\log^{-2}2$,$\log^{-1}2$, $(2\log^22)^{-1}$, $((1+\log2)\log 2)^{-1}$},
ymin = 0,
ymax = 2.2,
xmin = 0,
xmax=1.05,
height=6.1cm,
width=15cm,
legend style={at={(axis cs:1.01,1)},anchor=south west},
x label style={at={(axis description cs:0.9,-0.1)},anchor=north},
font = \small
]

\addplot [
name path=it,
domain=0:1, 
samples=100, 
style={ultra thick},
solid,
color=black,
]
{max(x/(ln(2)^2), (1-x)/ ln(2)};
\addlegendentry{$\minf/(n^\theta\log n)$}

\addplot [
name path=dd,
domain=0:1, 
samples=500, 
color=cyan,
style={ultra thick},
dashed
]
{max(x/(ln(2)^2), (1-x)/( ln(2)^2)};
\addlegendentry{$\malg/(n^\theta\log n)$}

\addplot [
name path=counting,
domain=0:1, 
samples=100, 
color=red,
style={ultra thick},
dashed
]
{(1-x)/(ln(2))};
\addlegendentry{$m_{\text{ad}}/(n^\theta\log n)$}

\addplot [
name path=border_up,
domain=0:1, 
samples=10, 
color=black,
style={thin},
draw opacity=0.00,
dashed
]
{1/(ln(2)*ln(2))+0.01};

\addplot [
name path=border_low,
domain=0:1, 
samples=10, 
color=black,
style={thin},
draw opacity=0.00,
dashed
]
{0.001};

\addplot [
        thick,
        color=white,
        fill=green, 
        fill opacity=0.15
    ]
    fill between[
        of=dd and it,
        soft clip={domain=0:1},
    ];

\addplot [
        thick,
        color=black,
        fill=yellow, 
        fill opacity=0.4
    ]
    fill between[
        of=it and counting,
        soft clip={domain=0.409:1},
    ];

\addplot [
        thick,
        color=blue,
        fill=blue, 
        fill opacity=0.2
    ]
    fill between[
        of=dd and border_up,
        soft clip={domain=0:1},
    ];

\addplot [
        thick,
        color=red,
        fill=red, 
        fill opacity=0.2
    ]
    fill between[
        of=it and border_low,
        soft clip={domain=0:0.409},
    ];

\addplot [
        thick,
        color=red,
        fill=red, 
        fill opacity=0.2
    ]
    fill between[
        of=counting and border_low,
        soft clip={domain=0.409:1},
    ];

\path[name path=axis] (axis cs:0,0) -- (axis cs:1,0);

\addplot+[
mark=none,
color=black,
dotted,
]
coordinates
{(0.5,0) (0.5,1.040684490)};
\addplot+[
mark=none,
color=black,
dotted,
]
coordinates
{(0.409,0) (0.409,0.8520789316)};
\addplot+[
mark=none,
color=black,
dotted,
]
coordinates
{(0,1.040684490) (0.5,1.040684490)};
\addplot+[
mark=none,
color=black,
dotted,
]
coordinates
{(0,0.8520789316) (0.409,0.8520789316)};

\end{axis}

\end{tikzpicture}
\end{minipage}

\caption{The phase transitions in group testing.
    The best previously known algorithm {\tt DD} succeeds in the blue but not in the green region.
    The new algorithm {\tt SPIV} succeeds in both the blue and the green region.
    The  black line indicates the non-adaptive information-theoretic threshold $\minf$, below which non-adaptive group testing is impossible.
	In the red area even (multi-stage) adaptive inference is impossible.
    Finally, the two-stage adaptive group testing algorithm from \Thm~\ref{Thm_ad} succeeds in the yellow region.}
\label{fig_bounds_illustration}
\end{figure}

}

\newcommand\FigSP{
\begin{figure}
\begin{tikzpicture}

\foreach \i in {1,...,12}
{
        \def\lab{x_\i};
        \node[circle,draw=black,fill=black,minimum size=1] (\lab) at (0.4*\i,0) {};
}
\foreach \i in {13,...,24}
{
        \def\lab{x_\i};
        \node[circle,draw=blue,fill=blue,minimum size=1] (\lab) at (0.4*\i,0) {};
}
\foreach \i in {25,...,36}
{
        \def\lab{x_\i};
        \node[circle,draw=black, fill=black,minimum size=1] (\lab) at (0.4*\i,0) {};
}
\foreach \i in {0,...,9}
{
        \def\x{4*\i};
        \draw[dashed] (0.4*\x+0.2,0.5) -- (0.4*\x+0.2,-2.5);
}
\foreach \i in {1,...,9}
{
        \def\labx{c_\i};
        \def\labaone{done_\i};
        \def\labatwo{dtwo_\i};
        \pgfmathsetmacro{\xcoord}{0.4*(\i+1.5)+3*0.4*(\i-1))};

        \coordinate (\labx) at (\xcoord,-0.3);
        \coordinate (\labaone) at (\xcoord-0.25,-1.4);
        \coordinate (\labatwo) at (\xcoord+0.25,-1.4);
}
\foreach \i in {13,...,15}
{
       \def\laba{a_\i};
        \pgfmathsetmacro{\xcoord}{0.4*\i+0.2};
        \node[rectangle, minimum size=8,draw=blue] (\laba) at (\xcoord,-1.7){};
}
\foreach \i in {17,...,19}
{
        \def\laba{a_\i};
        \pgfmathsetmacro{\xcoord}{0.4*\i+0.2};
       \node[rectangle, minimum size=8,draw=blue] (\laba) at (\xcoord,-1.7){};
}
\foreach \i in {21,...,23}
{
        \def\laba{a_\i};
        \pgfmathsetmacro{\xcoord}{0.4*\i+0.2};
       \node[rectangle, minimum size=8,draw=blue] (\laba) at (\xcoord,-1.7){};
}

\foreach \j in {2,3,4,8,9,10}{
    \pgfmathsetmacro{\ione}{2*\j-1};
    \pgfmathsetmacro{\itwo}{2*\j};
    \def\laba{a_\ione};
    \def\labb{a_\itwo};
    \pgfmathsetmacro{\xcoordone}{0.4*4*(\j-2)+1.2};
    \pgfmathsetmacro{\xcoordtwo}{0.4*4*(\j-2)+0.8};
    \node[rectangle, minimum size=8,draw=black] (\laba) at (\xcoordone,-1.7){};
    \node[rectangle, minimum size=8,draw=black] (\labb) at (\xcoordtwo,-1.7){};

}

\foreach \i in {4,5,6}
{
    \filldraw[fill=blue] (c_\i) -- (done_\i) -- (dtwo_\i) -- cycle;
}
\foreach \i in {1,2,3,7,8,9}
{
    \filldraw[fill=black] (c_\i) -- (done_\i) -- (dtwo_\i) -- cycle;
}

\foreach \i in {1,...,8}
{
    \pgfmathsetmacro{\x}{\i+1};
    \filldraw[fill=black!25] (c_\i) -- (done_\x) -- (dtwo_\x) -- cycle;
}
\foreach \i in {1,...,7}
{
    \pgfmathsetmacro{\x}{\i+2};
    \filldraw[fill=black!10] (c_\i) -- (done_\x) -- (dtwo_\x) -- cycle;
}

\filldraw[black!25] (0.2,-0.7) -- (1.2,-1.4)--(0.7,-1.4)--(0.2,-0.9)--cycle;
\filldraw[black!10] (0.2,-1.1) -- (1.2,-1.4)--(0.7,-1.4)--(0.2,-1.2)--cycle;
\filldraw[black!10] (0.2,-0.5) -- (2.8,-1.4)--(2.35,-1.4)--(0.2,-0.55)--cycle;
\filldraw[black!25] (13.8,-0.3) -- (14.6,-0.7)--(14.6,-0.9)--cycle;
\filldraw[black!10] (13.8,-0.3) -- (14.6,-0.5)--(14.6,-0.55)--cycle;
\filldraw[black!10] (12.2,-0.3) -- (14.6,-1.1)--(14.6,-1.2)--cycle;
\foreach \i in {1,2,5,6,9,10}
{
        \pgfmathsetmacro{\x}{\i-16};
        \pgfmathsetmacro{\xcoord}{28*0.4 + 0.4*\x+0.4};
        \node[rectangle, minimum size=8,draw=black] at (\xcoord,-2.1){};
}
\node (A) at (1,0.5) {$V[7]$};
\node (B) at (2.63,0.5) {$V[8]$};
\node (C) at (4.26,0.5) {$V[9]$};
\node[text=blue] (D) at (5.87,0.5) {$V[1]$};
\node[text=blue] (E) at (7.45,0.5) {$V[2]$};
\node[text=blue]  (F) at (9.05,0.5) {$V[3]$};
\node (G) at (10.65,0.5) {$V[4]$};
\node (H) at (12.28,0.5) {$V[5]$};
\node (I) at (13.91,0.5) {$V[6]$};

\node (J) at (1,-2.5) {$F[7]$};
\node (K) at (2.63,-2.5) {$F[8]$};
\node (L) at (4.26,-2.5) {$F[9]$};
\node (M) at (5.87,-2.5) {$F[1]$};
\node (N) at (7.45,-2.5) {$F[2]$};
\node (O) at (9.05,-2.5) {$F[3]$};
\node (P) at (10.65,-2.5) {$F[4]$};
\node (Q) at (12.28,-2.5) {$F[5]$};
\node (R) at (13.91,-2.5) {$F[6]$};
\node (S)[text=blue]  at (5.87,-3) {$F[0]$};
\node (S)[text=blue]  at (7.45,-3) {$F[0]$};
\node (S)[text=blue]  at (9.03,-3) {$F[0]$};
\node at (-0.2,-1){$\cdots$};
\node at (14.9,-1){$\cdots$};

\draw[->] (15.2, -1) -- (15.45, -1) -- (15.45, 0.9) -- (-0.75, 0.9) -- (-0.75, -1) -- (-0.45, -1);

\end{tikzpicture}
\caption[Idea]{The spatially coupled test design with $n = 36, \ell = 9, s = 3$. The individuals in the seed groups $V[1]\cup\cdots\cup V[s]$ (blue) are equipped with additional test $F[0]$ (blue rectangles).
The black rectangles represent the tests $F[1]\cup\cdots\cup F[\ell]$.}
\label{Fig_spatial_coupling_idea}

\end{figure}
}

\begin{document}
	
\title{Optimal group testing}
	
\thanks{Supported by DFG CO 646/3 and Stiftung Polytechnische Gesellschaft. An extended abstract version of this work has been submitted to the COLT 2020 conference.}

\author{Amin Coja-Oghlan, Oliver Gebhard, Max Hahn-Klimroth, Philipp Loick}
\address{Amin Coja-Oghlan, {\tt acoghlan@math.uni-frankfurt.de}, Goethe University, Mathematics Institute, 10 Robert Mayer St, Frankfurt 60325, Germany.}
\address{Oliver Gebhard, {\tt gebhard@math.uni-frankfurt.de}, Goethe University, Mathematics Institute, 10 Robert Mayer St, Frankfurt 60325, Germany.}
\address{Max Hahn-Klimroth, {\tt hahnklim@math.uni-frankfurt.de}, Goethe University, Mathematics Institute, 10 Robert Mayer St, Frankfurt 60325, Germany.}	
\address{Philipp Loick, {\tt loick@math.uni-frankfurt.de}, Goethe University, Mathematics Institute, 10 Robert Mayer St, Frankfurt 60325, Germany.}

\begin{abstract}%
 In the group testing problem the aim is to identify a small set of $k\sim n^\theta$ infected individuals out of a population size $n$, $0<\theta<1$.
We avail ourselves of a test procedure capable of testing groups of individuals, with the test returning a positive result iff at least one individual in the group is infected.
The aim is to devise a test design with as few tests as possible so that the set of infected individuals can be identified correctly with high probability.
We establish an explicit sharp information-theoretic/algorithmic phase transition $\minf$ for non-adaptive group testing, where all tests are conducted in parallel.
Thus, with more than $\minf$ tests the infected individuals can be identified in polynomial time \whp, while learning the set of infected individuals is information-theoretically impossible with fewer tests.
In addition, we develop an optimal adaptive scheme where the tests are conducted in two stages.
\hfill {\em MSc: 	05C80, 	60B20, 68P30}
\end{abstract}

\maketitle

\section{Introduction}

\subsection{Background and motivation.}

Various intriguing combinatorial problems come as inference tasks where we are to learn a hidden ground truth by means of indirect queries.
The goal is to get by with as small a number of queries as possible.
The ultimate solution to such a problem should consist of a positive algorithmic result showing that a certain number of queries suffice to learn the ground truth efficiently, complemented by a matching information-theoretic lower bound showing that with fewer queries the  problem is insoluble, regardless of computational resources.

Group testing is a prime example of such an inference problem~\cite{Aldridge_2019}.
The objective is to identify within a large population of size $n$ a subset of $k$ individuals infected with a rare disease.
We presume that the number of infected individuals scales as a power $k=\lceil n^\theta\rceil$ of the population size with an exponent $\theta\in(0,1)$,  a parametrisation suited to modelling the pivotal early stages of an epidemic~\cite{Wang}.
Indeed, since early on in an epidemic test kits might be in short supply, it is vital to get the most diagnostic power out the least number of tests.
To this end we assume that the test gear is capable of not merely testing a single individual but an entire group.
The test comes back positive if any one individual in the group is infected and negative otherwise.
While in {\em non-adaptive}  group testing all tests are conducted in parallel, in {\em adaptive} group testing test are conducted in several stages.
In either case we are free to allocate individuals to test groups as we please.
Randomisation is allowed.
What is the least number of tests required so that the set of infected individuals can be inferred from the test results with high probability?
Furthermore, in adaptive group testing, what is the smallest depth of test stages required?

Closing the considerable gaps that the best prior bounds left, the main results of this paper furnish matching algorithmic and information-theoretic bounds for both adaptive and non-adaptive group testing.
Specifically, the best prior information-theoretic lower bound derives from the following folklore observation.
Suppose that we conduct $m$ tests that each return either `positive' or `negative'.
Then to correctly identify the set of infected individuals we need the total number $2^m$ of conceivable test results to asymptotically exceed the number $\binom nk$ of possible sets of infected individuals.
Hence,  $2^m\geq(1+o(1))\binom nk$.
Thus, Stirling's formula yields the lower bound
\begin{align}\label{eqInfAdapt} 
\madapt&=\frac{1-\theta}{\log2}n^\theta \log n,
\end{align}
which applies to both adaptive and non-adaptive testing.
On the positive side, a randomised non-adaptive test design with
\begin{align}\label{eqDD}
\malg\sim\frac{\max \cbc{\theta,1-\theta}}{\ln^2 2}  n^\theta \log n
\end{align}
tests exists from which a greedy algorithm called {\tt DD} correctly infers the set of infected individuals \whp~\cite{Johnson_2019}.
Clearly, $\madapt<\malg$ for all infection densities $\theta$ and $\malg/\madapt\to\infty$ as $\theta\to1$.
In addition, there is an efficient adaptive three-stage group testing scheme that asymptotically matches the lower bound $\madapt$~\cite{Scarlett_2019}.

We proceed to state the main results of the paper.
First, improving both the information-theoretic and the algorithmic bounds, we present optimal results for non-adaptive group testing.
Subsequently we show how the non-adaptive result can be harnessed to perform adaptive group testing with the least possible number $(1+o(1))\madapt$ of tests in only two stages.

\subsection{Non-adaptive group testing}\label{Sec_intro_non}
A {\em non-adaptive test design} is a bipartite graph $G=(V\cup F,E)$ with one vertex class $V=V_n=\{x_1,\ldots,x_n\}$ representing individuals and the other class $F=F_m=\{a_1,\ldots,a_m\}$ representing tests.
For a vertex $v$ of $G$ denote by $\partial v=\partial_Gv$ the set of neighbours of $v$.
Thus, an individual $x_j$ takes part in a test $a_i$ iff $x_j\in\partial a_i$.
Since we can shuffle the individuals randomly, we may safely assume that the vector $\SIGMA\in\cbc{0,1}^{V}$ whose $1$-entries mark the infected individuals is a uniformly random vector of Hamming weight $k$.
Furthermore, the test results induced by $\SIGMA$ read
\begin{align*}
\hat\SIGMA_{a_i}=\hat\SIGMA_{G,a_i}&=\max_{x\in\partial a_i}\SIGMA_x.
\end{align*}
Hence, given $\hat\SIGMA=\hat\SIGMA_{G}=(\hat\SIGMA_{G,a})_{a\in F}$ and $G$ we aim to infer $\SIGMA$.
Thus, we can represent an inference procedure by a function $\cA_G:\{0,1\}^m\to\{0,1\}^n$.
The following theorem  improves the lower bound on the number of tests required for successful inference.
Let
\begin{align}\label{eqInfUpper}
\minf=\minf(n,\theta)&=\max \cbc{\frac{\theta}{\ln^2 2},\frac{1-\theta}{\log2}}  n^\theta \log n.
\end{align}

\begin{theorem}\label{opt}
For any $0<\theta<1$, $\eps>0$ there exists $n_0=n_0(\theta,\eps)$ such that for all $n>n_0$,
all test designs $G$ with $m\leq(1-\eps)\minf$ tests and for every function $\cA_G:\{0,1\}^m\to\{0,1\}^n$ we have
\begin{align}\label{eqopt1}
\pr\brk{\cA_G(\hat\SIGMA_G)=\SIGMA}<\eps.
\end{align}
\end{theorem}

\Thm~\ref{opt} rules out both deterministic and randomised test designs and inference procedures because \eqref{eqopt1} holds uniformly for all $G$ and all $\cA_G$.
Thus, no test design, randomised or not, with fewer than $\minf$ tests allows to infer the set of infected individuals with a non-vanishing probability.
Since $\minf$ matches $\malg$ from \eqref{eqDD} for $\theta\geq1/2$, \Thm~\ref{opt} shows that the positive result from~\cite{Johnson_2019} is optimal in this regime.
The following theorem closes the remaining gap by furnishing an optimal positive result for all $\theta$.

\begin{theorem} \label{thm_SC}
For any $0<\theta<1$, $\eps>0$ there is $n_0=n_0(\theta,\eps)$ such that for every $n>n_0$ there exist a randomised test design $\G$ comprising $m\leq(1+\eps)\minf$  tests and a polynomial time algorithm \SPIV\ that given $\G$ and the test results $\hat\SIGMA_{\G}$ outputs $\SIGMA$ \whp{}
\end{theorem}

An obvious candidate for an optimal test design appears be a plain random bipartite graph.
In fact, prior to the present work the best known test design consisted of a uniformly random bipartite graph where all vertices in $V_n$ have the same degree $\Delta$.
In other words, every individual independently joins $\Delta$ random test groups.
Applied to this random $\Delta$-out test design the {\tt DD} algorithm correctly recovers the set of infected individuals in polynomial time provided that the number of tests exceeds $\malg$ from \eqref{eqDD}.
However, $\malg$ strictly exceeds $\minf$ for $\theta<1/2$.
While the random $\Delta$-out test design with $(1+o(1))\minf$ tests is known to admit an exponential time algorithm that successfully infers the set of infected individuals \whp~\cite{Coja_2019}, we do not know of a polynomial time that solves this inference problem.
Instead, to facilitate the new efficient inference algorithm \SPIV\ the test design for \Thm~\ref{thm_SC} relies on a blend of a geometric and a random construction that is inspired by recent advances in coding theory known as spatially coupled low-density parity check codes~\cite{Felstrom_1999, Kudekar_2011}.

Finally, for
\begin{align}\label{eqthetaopt}
\theta\leq\frac{\log2}{1+\log2}\approx0.41
\end{align}
the number $\minf$ of tests required by \Thm~\ref{thm_SC} matches the folklore lower bound $\madapt$ from \eqref{eqDD} that applies to both adaptive and non-adaptive group testing.
Hence, in this regime  adaptivity confers no advantage.
By contrast, for $\theta>\log(2)/(1+\log2)$ the adaptive bound $\madapt$ is strictly smaller than $\minf$.
Consequently, in this regime at least two test stages are necessary to match the lower bound.
Indeed, the next theorem shows that two stages suffice.

\FigG

\subsection{Adaptive group testing}
A {\em two-stage test design} consists of a bipartite graph $G=(V,F)$ along with a second bipartite graph $G'=G'(G,\hat\SIGMA_G)=(V',F')$ with $V'\subset V$ that may depend on the tests results $\hat\SIGMA_G$ of the first test design $G$.
Hence, the task is to learn $\SIGMA$ correctly \whp{} from $G,\hat\SIGMA_G,G'$ and the test results $\hat\SIGMA_{G'}$ from the second stage while minimising the total number $|F|+|F'|$ of tests.
The following theorem shows that a two-stage test design and an efficient inference algorithm exist that meet the multi-stage adaptive lower bound \eqref{eqInfAdapt}.

\begin{theorem}\label{Thm_ad}
For any $0<\theta<1$, $\eps>0$ there is $n_0=n_0(\theta,\eps)$ such that for every $n>n_0$ there exist a two-stage test design with no more than $(1+\eps)\madapt$ tests in total and a polynomial time inference algorithm that outputs $\SIGMA$ with high probability.
\end{theorem}

\noindent
\Thm~\ref{Thm_ad} improves over \cite{Scarlett_2019} by reducing the number of stages from three to two, thus potentially significantly reducing the overall time required to complete the test procedure~\cite{Chen_2008, Kwang_2006}.
The proof of \Thm~\ref{Thm_ad} combines the test design and efficient algorithm from \Thm~\ref{thm_SC} with ideas from~\cite{Scarlett_2018}.

The question of whether an `adaptivity gap' exists for group testing, i.e., if the number of tests can be reduced by allowing multiple stages, has been raised prominently~\cite{Aldridge_2019}. 
\Thm s~\ref{opt}--\ref{Thm_ad} answer this question comprehensively.
While for $\theta\leq\ln(2)/(1+\ln(2))\approx0.41$ adaptivity confers no advantage, \Thm~\ref{opt} shows that for $\theta>\ln(2)/(1+\ln(2))$ there is a widening gap between $\madapt$ and the number $\minf$ of tests required by the optimal non-adaptive test design.
Further, \Thm~\ref{Thm_ad} demonstrates that this gap can be closed by allowing merely two stages.
Figure~\ref{fig_bounds_illustration} illustrates the thresholds from \Thm s~\ref{opt}--\ref{Thm_ad}.

\subsection{Discussion}\label{Sec_related}

The group testing problem was first raised in 1943, when Dorfman \cite{Dorfman_1943} proposed a two-stage adaptive test design to test the US Army for syphilis: in a first stage disjoint groups of equal size are tested.
All members of negative test groups are definitely uninfected.
Then, in the second stage the members of positive test groups get tested individually.
Of course, this test design is far from optimal, but Dorfman's contribution triggered attempts at devising improved test schemes.

At first combinatorial group testing, where the aim is to construct a test design that is guaranteed to succeed on {\em all} vectors $\SIGMA$, attracted significant attention.
This version of the problem was studied, among others, by \Erdos\ and \Renyi~\cite{Erdos_1963}, D'yachkov and Rykov~\cite{Dyachkov_1982} and Kautz and Singleton~\cite{Kautz_1964}.
Hwang~\cite{Hwang_1972} was the first to propose an adaptive test design that asymptotically meets the information-theoretic lower bound $\madapt$ from \eqref{eqInfAdapt}  for all $\theta \in [0,1]$.
However, this test design requires an unbounded number of stages.
Conversely, D{'}yachkov and Rykov~\cite{Dyachkov_1982} showed that $\madapt$ tests do not suffice for non-adaptive group testing.
Indeed, $m \geq \min \cbc{\Omega(k^2), n}$ tests are required non-adaptively, making individual testing optimal for $\theta>1/2$.
For an excellent survey of combinatorial group testing see~\cite{Aldridge_2019}.

Since the early 2000s attention has shifted to probabilistic group testing, which we study here as well.
Thus, instead of asking for test designs and algorithms that are guaranteed to work for {\em all} $\SIGMA$, we are content with recovering $\SIGMA$ with high probability. 
Berger and Levenshtein~\cite{Berger_2002} presented a two-stage probabilistic group testing design and algorithm requiring 
\begin{align*}
\mBL \sim 4n^\theta \log n
\end{align*}
tests in expectation.
Their test design, known as the Bernoulli design, is based on a random bipartite graph where each individual joins every test independently with a carefully chosen edge probability.
For a fixed $\theta$ the number $\mBL$ of tests is within a bounded factor of the information-theoretic lower bound $\madapt$ from \eqref{eqInfAdapt}, although the gap $\madapt/\mBL$ diverges as $\theta\to1$.
Unsurprisingly, the work of Berger and Levenshtein spurred efforts at closing the gap.
\Mezard, Tarzia and Toninelli proposed a different two-stage test design whose first stage consists of a random bipartite graph called the constant weight design~\cite{Mezard_2008}.
Here each individual independently joins an equal number of random tests.
For their two-stage design they obtained an inference algorithm that gets by with about
\begin{align} \label{eq_Mezard}
\mMT \sim \frac{1-\theta}{\log^2 2} n^\theta \log n.
\end{align}
tests, a factor of $1/\log 2$ above the elementary bound $\madapt$.
Conversely, \Mezard, Tarzia and Toninelli showed by means of the FKG inequality and positive correlation arguments that two-stage test algorithms from a certain restricted class cannot beat the bound~\eqref{eq_Mezard}.
Furthermore, Aldridge, Johnson and Scarlett analysed non-adaptive test designs and inference algorithms~\cite{Aldridge_2014, Johnson_2019}.
For the Bernoulli test design their best efficient algorithm {\tt DD} requires
\begin{align*}
\mDDB \sim \eul\cdot\max \cbc{\theta, 1-\theta} n^\theta \log n.
\end{align*}
tests.
For the constant weight design they obtained the bound $\malg$ from~\eqref{eqDD}.
In addition, in a previous article~\cite{Coja_2019} we showed that on the constant weight design an exponential time algorithm correctly identifies the set of infected individuals \whp\ if the number of tests exceeds $\minf$ from~\eqref{eqInfUpper}.
Furthermore, Scarlett~\cite{Scarlett_2019} discovered the aforementioned three-stage test design and polynomial time algorithm that matches the universal lower bound $\madapt$ from \eqref{eqInfAdapt}.
Finally, concerning lower bounds, in the case of a linear number $k=\Theta(n)$ infected individuals Aldridge~\cite{Aldridge_2018} showed via arguments similar to~\cite{Mezard_2008} that individual testing is optimal in the non-adaptive case, while Ungar~\cite{Ungar_1960} proved that individual testing is optimal even adaptively once $k\geq(3-\sqrt5)n/2$.

A further variant of group testing is known as the quantitative group testing or the coin weighing problem.
In this problem tests are assumed to not merely indicate the presence of at least one infected individual but to return the number of infected individuals.
Thus, the tests are significantly more powerful.
For quantitative group testing with $k$ infected individuals Alaoui, Ramdas, Krzakala, Zdeborov\'a and~Jordan~\cite{Lenka_pooled_data2} presented a test design 
with
\begin{align*}
 \mQGT& \sim 2\bc{1+\frac{(n-k)\log(1-k/n)}{k \log(k/n)}} \frac{k \log(n/k)}{\log(k)}&\mbox{for}&&k&=\Theta(n)
\end{align*}
tests from which the set of infected individuals can be inferred in exponential time;
    the paper actually deals with the slightly more general pooled data problem.
However, no efficient algorithm is known to come within a constant factor of $\mQGT$.
Indeed, the best efficient algorithm, due to the same authors~\cite{Lenka_pooled_data}, 
requires $\Omega(k\log(n/k))$ tests.

More broadly, the idea of harnessing random graphs to tackle inference problems has been gaining momentum.
One important success has been the development of capacity achieving linear codes called spatially coupled low-density parity check (`LDPC') codes~\cite{Kudekar_2011,Kudekar_2013}.
The Tanner graphs of these codes, which represent their check matrices, consist of a linear sequence of sparse random bipartite graphs with one class of vertices corresponding to the bits of the codeword and the other class corresponding to the parity checks.
The bits and the checks are divided equitably into a number of compartments, which are arranged along a line.
Each bit of the codeword takes part in random checks in a small number of preceding and subsequent compartments of checks along the line.
This combination of a spatial arrangement and randomness facilitates efficient decoding by means of the Belief Propagation message passing algorithm.
Furthermore, the general design idea of combining a linear spatial structure with a random graph has been extended to other inference problems.
Perhaps the most prominent example is compressed sensing, i.e., solving an underdetermined linear system subject to a sparsity constraint~\cite{ Donoho_2006,Donoho_2013,Krzakala_2012,Kudekar_2010_2}, where a variant of Belief Propagation called Approximate Message Passing matches an information-theoretic lower bound from~\cite{Wu_2010}.

While in some inference problems such as LDPC decoding or compressed sensing the number of queries required to enable an efficient inference algorithm matches the information-theoretic lower bound, in many other problems gaps remain.
A prominent example is the stochastic block model~\cite{Abbe_2017,Decelle,Moore}, an extreme case of which is the notorious planted clique problem~\cite{Alon_1998}.
For both these models the existence of a genuine computationally intractable phase where the problem can be solved in exponential but not in polynomial time appears to be an intriguing possibility.
Further examples include code division multiple access~\cite{Takeuchi_2011, Zdeborova_2016}, quantitative group testing~\cite{Lenka_pooled_data}, sparse principal component analysis~\cite{Brennan_2019} and sparse high-dimensional regression~\cite{Reeves_2019}.
The problem of solving the group testing inference problem on the test design from~\cite{Johnson_2019} could be added to the list.
Indeed, while an exponential time algorithm (that reduces the problem to minimum hypergraph vertex cover) infers the set of infected individuals \whp\ with only $(1+\eps)\minf$ tests, the best known polynomial algorithm requires $(1+\eps)\malg$ tests.

Instead of developing a better algorithm for the test design from~\cite{Johnson_2019}, here we exercise the discretion  of constructing a different test design  that the group testing problem affords.
The new design is tailored to enable an efficient algorithm \SPIV\ for \Thm~\ref{thm_SC} that gets by with $(1+\eps)\minf$ tests.
While prior applications of the idea of spatial coupling such as coding and compressed sensing required sophisticated message passing algorithms~\cite{Felstrom_1999, Kudekar_2011, Kudekar_2013}, the \SPIV\ algorithm is purely combinatorial and extremely transparent.
The main step of the algorithm merely computes a weighted sum to discriminate between infected individuals and `disguised' healthy individuals. Furthermore, the analysis of the algorithm is based on a technically subtle but conceptually clean large deviations analysis.
This technique of blending combinatorial ideas and large deviations methods with spatial coupling promises to be an exciting route for future research. 
Applications might include noisy versions of group testing, quantitative group testing or the coin weighing problem \cite{Lenka_pooled_data}.
Beyond these immediate extensions, it would be most interesting to see if the \SPIV\ strategy extends to other inference problems for sparse data.

\subsection{Organisation}
After collecting some preliminaries and introducing notation in \Sec~\ref{Sec_pre},  we prove \Thm~\ref{opt} in \Sec~\ref{xsec_opt}.
\Sec~\ref{Sec_alg} then deals with the test design and the inference algorithm for \Thm~\ref{thm_SC}.
Finally, in \Sec~\ref{sec_prop_adaptive} we prove \Thm~\ref{Thm_ad}.

\section{Preliminaries}\label{Sec_pre}

\noindent
As we saw in \Sec~\ref{Sec_intro_non} a non-adaptive test design can be represented by a bipartite graph $G=(V\cup F,E)$ with one vertex class $V$ representing the individuals and the other class $F$ representing the tests.
We refer to the number $|V|$ of individuals as the {\em order} of the test design and to the number $|F|$ of tests as its {\em size}.
For a vertex $v$ of $G$ we denote by $\partial_Gv$ the set of neighbours.
Where $G$ is apparent from the notation we just write $\partial v$.
Furthermore, for an integer $k\leq|V|$ we denote by $\SIGMA_{G,k}=(\SIGMA_{G,k,x})_{x\in V}\in\{0,1\}^V$ a random vector of Hamming weight $k$.
Additionally, we let
\begin{align}\label{eqSigma}
\hat\SIGMA_{G,k}&=(\hat\SIGMA_{G,k,a})_{a\in F}\in\{0,1\}^F&\mbox{with}&&\hat\SIGMA_{G,k,a}=\max_{x\in\partial_Ga}\SIGMA_{G,k,x}
\end{align}
be the associated vector of test results.
Where $G$ and/or $k$ are apparent from the context, we drop them from the notation.
More generally, for a given vector $\tau\in\{0,1\}^V$ we introduce a vector $\hat\tau_G=(\hat\tau_{G,a})_{a\in F}$ by letting $\hat\tau_{G,a}=\max_{x\in\partial_Ga}\tau_{x}$, just as in \eqref{eqSigma}.
Furthermore, for a given $\tau\in\{0,1\}^V$ we let
\begin{align*}
V_0(G,\tau)&=\cbc{x\in V:\tau_x=0},&
V_1(G,\tau)&=\cbc{x\in V:\tau_x=1},&
F_0(G,\tau)&=\cbc{a\in F:\hat\tau_{G,a}=0},&
F_1(G,\tau)&=\cbc{a\in F:\hat\tau_{G,a}=1}.
\end{align*}

The {\em Kullback-Leibler divergence} of $p,q\in(0,1)$ is denoted by
\begin{align*}
    \KL{q}{p} = q \log \bc{\frac{q}{p}} + (1-q) \log \bc{\frac{1-q}{1-p}}.
\end{align*}
We will  occasionally apply the following Chernoff bound.

\begin{lemma}[\cite{Janson_2011}] \label{lem_chernoff}
Let $\vX$ be a binomial random variable with parameters $N,p$.
Then
\begin{align}\label{eqChernoff1}
    \Pr\brk{X \geq {qN}} &\leq \exp \bc{-N\KL{q}{p}} \quad \text{for $p<q<1$,} \\
    \Pr\brk{X \leq {qN}} &\leq \exp \bc{-N\KL{q}{p}} \quad \text{for $0<q<p$.}\label{eqChernoff2}
\end{align}
\end{lemma}

In addition, we recall that the {\em hypergeometric distribution} $\Hyp(L,M,N)$ is defined by
\begin{align*}
\pr\brk{\Hyp(L,M,N)=k}&=\binom Mk\binom{L-M}{N-k}\binom LN^{-1}.&&(k\in\{0,1,\ldots,M\wedge N\}).
\end{align*}
Hence, out of a total of $L$ items of which $M$ are special we draw $N$ items without replacement and count the number of special items in the draw.
The mean of the hypergeometric distribution equals $MN/L$.
It is well known that the Chernoff bound extends to the hypergeometric distribution.

\begin{lemma}[\cite{Hoeffding}] \label{lem_hyperchernoff}
For a  hypergeometric variable $\vX\sim\Hyp(L,M,N)$ the bounds \eqref{eqChernoff1}--\eqref{eqChernoff2} hold with $p=M/L$.
\end{lemma}

Throughout the paper we use asymptotic notation $o(\nix),\omega(\nix),O(\nix),\Omega(\nix),\Theta(\nix)$ to refer to limit $n\to\infty$.
It is understood that the constants hidden in, e.g., a $O(\nix)$-term may depend on the density parameter $\theta$ or other parameters.

\section{The information theoretic lower bound}\label{xsec_opt}

\noindent
In this section we prove \Thm~\ref{opt}.
The proof combines techniques based on the FKG inequality and positive correlation that were developed in~\cite{Aldridge_2019,Mezard_2008} with new combinatorial ideas.
Throughout this section we fix a number $\theta\in(0,1)$ and we let $k=\lceil n^\theta\rceil$.

\subsection{Outline}
The starting point is a simple and well known observation.
Namely, for a test design $G=G_{n,m}=(V_n,F_m)$ and a vector $\tau\in\{0,1\}^{F_m}$ of test results let
\begin{align*}
\cS_k(G,\tau)&=\cbc{\sigma\in\cbc{0,1}^{V_n}:\sum_{x\in V_n}\sigma_x=k,\ \hat\sigma_G=\tau}
\end{align*}
be the set of all possible vectors $\sigma$ of Hamming weight $k$ that give rise to the test results $\tau$.
Further, let $Z_k(G,\tau)=|\cS_k(G,\tau)|$ be the number of such vectors $\sigma$.
Also recall that $\SIGMA=\SIGMA_{G,k}\in\{0,1\}^{V_n}$ is a random vector of Hamming weight $k$ and that $\hat\SIGMA=\hat\SIGMA_{G,k}$ comprises the test results that $\SIGMA$ renders under the test design $G$.
We observe that the posterior of $\SIGMA$ given $\hat\SIGMA$ is the uniform distribution on $\cS_k(G,\hat\SIGMA)$.

\begin{fact}\label{Prop_Nishi}
For any $G$, $\sigma\in\{0,1\}^{V_n}$ we have $\pr\brk{\SIGMA=\sigma\mid\hat\SIGMA}=\vecone\cbc{\sigma\in\cS_k(G,\hat\SIGMA)}/Z_k(G,\hat\SIGMA)$.
\end{fact}

\noindent
As an immediate consequence of Fact~\ref{Prop_Nishi}, the success probability of any inference scheme $\cA_G:\cbc{0,1}^{F_m}\to\cbc{0,1}^{V_n}$ is bounded by $1/Z_k(G,\hat\SIGMA)$.
Indeed, an optimal  inference algorithm is to simply return a uniform sample from $\cS_k(G,\hat\SIGMA)$.

\begin{fact}\label{Prop_Nishi2}
For any test design $G$ and for any $\cA_G:\cbc{0,1}^{F_m}\to\cbc{0,1}^{V_n}$ we have
$\pr\brk{\cA_G(\hat\SIGMA)=\SIGMA\mid\hat\SIGMA}\leq 1/Z_k(G,\hat\SIGMA).$
\end{fact}

\noindent
Hence, in order to prove \Thm~\ref{opt} we just need to show that $Z_k(G,\hat\SIGMA)$ is large for any test design $G$ with $m<(1-\eps)\minf$ tests.
In other words, we need to show that \whp{} there are many vectors $\sigma\in\cS_k(G,\hat\SIGMA)$ that give rise to the test results $\hat\SIGMA$.

We obtain these $\sigma$ by making diligent local changes to $\SIGMA$.
More precisely, we identify two sets $V_{0+}=V_{0+}(G,\SIGMA)$, $V_{1+}=V_{1+}(G,\SIGMA)$ of individuals whose infection status can be flipped without altering the test results.
Specifically, following~\cite{Aldridge_2018} we call an individual $x\in V_n$ {\em disguised} if every test $a\in\partial_Gx$ contains another individual $y\in\partial_Ga\setminus\cbc x$ with $\SIGMA_y=1$.
Let $V_+=V_+(G,\SIGMA)$ be the set of all disguised individuals.
Moreover, let
\begin{align}\label{eqdisguise}
V_{0+}=V_{0+}(G,\SIGMA)&=\cbc{x\in V_+:\SIGMA_x=0},&V_{1+}=V_{1+}(G,\SIGMA)&=\cbc{x\in  V_+:\SIGMA_x=1}.
\end{align}
Hence, $V_{0+}$ is the set of all healthy disguised individuals while $V_{1+}$ contains all infected disguised individuals.

\begin{fact}\label{Prop_Nishi3}
We have $Z_k(G,\hat\SIGMA)\geq|V_{0+}(G,\SIGMA)|\cdot|V_{1+}(G,\SIGMA)|$.
\end{fact}
\begin{proof}
For a pair $(x,y)\in V_{0+}(G,\SIGMA)\times V_{1+}(G,\SIGMA)$ obtain $\TAU$ from $\SIGMA$ by letting $\TAU_x=1,\TAU_y=0$ and $\TAU_z=\SIGMA_z$ for all $z\neq x,y$.
Then $\TAU$ has Hamming weight $k$ and $\hat\TAU_G=\hat\SIGMA$.
Thus, $\TAU\in\cS_k(G,\hat\SIGMA)$.
\end{proof}

Hence, an obvious proof strategy for \Thm~\ref{opt} is to exhibit a large number of disguised individuals.
A similar strategy has been pursued in the proof of the conditional lower bound of \Mezard, Tarzia and Toninelli~\cite{Mezard_2008} and the proof of Aldridge's lower bound for the linear case $k=\Theta(n)$~\cite{Aldridge_2018}.
Both~\cite{Aldridge_2018,Mezard_2008} exhibit disguised individuals via positive correlation and the FKG inequality.
However, we do not see how to stretch such arguments to obtain the desired lower bound for all $\theta\in(0,1)$.
Yet for $\theta$ {\em extremely} close to one it is possible to combine the positive correlation argument with new combinatorial ideas to obtain the following.

\begin{proposition}\label{Prop_large_theta}
For any $\eps>0$ there exists $\theta_0=\theta_0(\eps)<1$ such that for every $\theta\in(\theta_0,1)$ there exists $n_0=n_0(\theta,\eps)$ such that for all $n>n_0$ and all test designs $G=G_{n,m}$ with $m\leq(1-\eps)\minf$ we have
$$\pr\brk{|V_{0+}(G,\SIGMA)|\wedge|V_{1+}(G,\SIGMA)|\geq\log n}>1-\eps.$$
\end{proposition}

\noindent
The proof of \Prop~\ref{Prop_large_theta} can be found in \Sec~\ref{Sec_Prop_large_theta}.

The second step towards \Thm~\ref{opt} is a reduction from larger to smaller values of $\theta$.
Suppose we wish to apply a test scheme designed for an infection density $\theta\in(0,1)$ to a larger infection density $\theta'\in(\theta,1)$.
Then we could dilute the larger infection density by adding a large number of healthy `dummy' individuals.
A careful analysis of this dilution process yields the following result.
Due to the elementary lower bound \eqref{eqInfAdapt} we need not worry about $\theta\leq\log(2)/(1+\log 2)$.

\begin{proposition}\label{Prop_small_theta}
For any $\log(2)/(1+\log(2))<\theta<\theta'<1$, $t>0$ there exists $n_0=n_0(\theta,\theta',t)>0$ such that for every $n>n_0$ and
for every test design $G$ of order $n$ there exist an integer $n'$ such that
$$k=\lceil n^{\theta}\rceil=\lceil n'\,^{\theta'}\rceil$$
and a test design $G'$ of order $n'$ with the same number of tests as $G$ such that the following is true.
Let $\TAU\in\{0,1\}^{V_{n'}}$ be a random vector of Hamming weight $k$ and let $\hat\TAU_a=\max_{x\in\partial_{G'}a}\TAU_x$ comprise the tests results of $G'$.
Then
$$\pr\brk{Z_k(G,\hat\SIGMA)\leq t}\leq \pr\brk{Z_k(G',\hat\TAU)\leq t}.$$
\end{proposition}

\noindent
Hence, if a test design exists for $\theta<\theta'$ that beats $\minf(n,\theta)$, then there is a test design for infection density $\theta'$ that beats $\minf(n',\theta')$. 
We prove \Prop~\ref{Prop_large_theta} in \Sec~\ref{Sec_Prop_large_theta}.
\Thm~\ref{opt} is an easy consequence of \Prop s~\ref{Prop_large_theta} and \ref{Prop_small_theta}.

\begin{proof}[Proof of \Thm~\ref{opt}]
For $\theta\leq\log(2)/(1+\log(2))$ the assertion follows from the elementary lower bound \eqref{eqInfAdapt}.
Hence, fix $\eps>0$ and assume for contradiction that some $\theta\in(\log(2)/(1+\log(2)),1)$ for infinitely many $n$ admits a test design $G$ of order $n$ and size $m\leq(1-\eps)\minf(n,\theta)$  such that $\pr\brk{Z_k(G,\hat\SIGMA_G)\leq t}\geq \eps$.
Then \Prop~\ref{Prop_small_theta} shows that for $\theta'>\theta$ arbitrarily close to one for an integer $n'$ with $k=\lceil n'\,^{\theta'}\rceil$ a test design $G'=G_{n',m}$ exists such that
\begin{align}\label{eqThmOpt1}
\pr\brk{Z_k(G',\hat\TAU)\leq 1/\eps}\geq\eps.
\end{align}
Furthermore, \eqref{eqInfUpper} shows that  for large $n$,
\begin{align*}
\minf(n',\theta')=\frac{\theta'}{\ln^2 2} n'\,^{\theta'}\log n'=\frac{\theta+o(1)}{\ln^2 2} n^{\theta}\log n=(1+o(1))\minf(n,\theta).
\end{align*}
Hence, the number $m$ of tests of $G'$ satisfies $m\leq(1-\eps+o(1))\minf(n',\theta')$.
Thus, \eqref{eqThmOpt1} contradicts Fact~\ref{Prop_Nishi3} and \Prop~\ref{Prop_large_theta}.
\end{proof}

\subsection{Proof of \Prop~\ref{Prop_large_theta}}\label{Sec_Prop_large_theta}
Given a small $\eps>0$ we choose $\theta_0=\theta_0(\eps)\in(0,1)$ sufficiently close to one and fix $\theta\in(\theta_0,1)$.
Additionally, pick $\xi=\xi(\eps,\theta)\in(0,1)$ such that
\begin{align}\label{eqxietatheta}
2(1-\theta)<\xi<\theta\eps.
\end{align}
We fix $\eps,\theta,\xi$ throughout this section.

To avoid the (mild) stochastic dependencies that result from the total number of infected individuals being fixed, instead of $\SIGMA$ we will consider a vector $\CHI\in\{0,1\}^{V_n}$ whose entries are stochastically independent.
Specifically, every entry of $\CHI$ equals one with probability
\begin{align*}
p&=\frac{k-\sqrt k\log n}n
\end{align*}
independently.
Let $\hat\CHI_G\in\{0,1\}^{F_m}$ be the corresponding vector of test results.
The following lemma shows that it suffices to estimate  $|\zeroplus(G,\CHI)|,|\oneplus(G,\CHI)|$.
Let $G$ denote an arbitrary test design with individuals $V_n=\{x_1,\ldots,x_n\}$ and tests $F_m=\{a_1,\ldots,a_m\}$.

\begin{lemma}\label{Lemma_bin}
There is $n_0=n_0(\theta,\eps)$ such that for all $n>n_0$ and for all  $G$ with $m\leq\minf$ the following is true:
$$\mbox{if $\pr\brk{|V_{0+}(G,\CHI)|\wedge|V_{1+}(G,\CHI)|\geq2\log n}>1-\eps/4$, then
$\pr\brk{|\zeroplus(G,\SIGMA)|\wedge |\oneplus(G,\SIGMA)|\geq\log n}>1-\eps.$}$$
\end{lemma}
\begin{proof}
Let $\cX=\{k-2\sqrt k\log n\leq\sum_{x\in V_n}\CHI_x\leq k\}$.
The Chernoff bound shows for large enough $n$,
\begin{align}\label{eqLemma_bin1}
\pr\brk\cX>1-\eta/4.
\end{align}
Further, given $\cX$ we can couple $\CHI$, $\SIGMA$ such that the latter is obtained by turning $k-\sum_{x\in V_n}\CHI_x$ random zero entries of the former into ones.
Since turning zero entries into ones can only increase the number of disguised individuals, on $\cX$ we have
\begin{align}\label{eqLemma_bin2}
V_{1+}(G,\SIGMA)&\geq V_{1+}(G,\CHI).
\end{align}
Of course, it is possible that $|V_{0+}(G,\SIGMA)|<|V_{0+}(G,\CHI)|$.
But since on $\cX$ the two vectors $\SIGMA,\CHI$ differ in no more than $2\sqrt k\log n$ entries, we obtain the bound
\begin{align*}
\Erw\brk{|V_{0+}(G,\CHI)|-|V_{0+}(G,\SIGMA)|\mid\cX}\leq\frac{2\sqrt k\log n}{n-k}|V_{0+}(G,\CHI)|<n^{-1/3}|V_{0+}(G,\CHI)|,
\end{align*}
provided $n$ is sufficiently large.
Hence, Markov's inequality shows that for large enough $n$,
\begin{align}\label{eqLemma_bin3}
\pr\brk{|V_{0+}(G,\CHI)|-|V_{0+}(G,\SIGMA)|> |V_{0+}(G,\CHI)|/2\mid\cX}<\eps/4.
\end{align}
Combining \eqref{eqLemma_bin1}, \eqref{eqLemma_bin2} and \eqref{eqLemma_bin3} completes the proof.
\end{proof}

As a next step we show that there is no point in having very big tests $a$ that contain more than, say, $\Gamma=\Gamma(n,\theta)=n^{1-\theta}\log n$ individuals.
This is because anyway all such tests are positive \whp, so there is little point in actually conducting them.
Indeed, the following lemma shows that \whp\ all tests of very high degree contain at least two infected individuals.

\begin{lemma}\label{Prop_removeHighDegs}
There exists $n_0=n_0(\theta,\eps)>0$ such that for all $n>n_0$ and all test designs $G$ with $m\leq \minf$ tests,
\begin{align*}
\pr\brk{\exists a\in F_m:|\partial_Ga|>\Gamma\wedge|\partial_Ga\cap V_1(G,\CHI)|\leq1}<\eps/8.
\end{align*}
\end{lemma}
\begin{proof}
Consider a test $a$ of degree $\gamma=|\partial_Ga|\geq\Gamma$.
Because in $\CHI$ each of the $\gamma$ individuals that take part in $a$ is infected with probability $p$ independently, we have
\begin{align}\label{eqProp_removeHighDegs1}
\pr\brk{|\partial_Ga\cap V_1(G,\SIGMA)|\leq1}&
=\pr\brk{\Bin(\gamma,p)\leq1}=(1-p)^\gamma+\gamma p(1-p)^{\gamma-1}\leq(1+\gamma p/(1-p))\exp(-\gamma p)=n^{o(1)-1}.
\end{align}
Since $m\leq\minf=O(n^\theta)$ for a fixed $\theta<1$, the assertion follows from \eqref{eqProp_removeHighDegs1} and the union bound.
\end{proof}

Let $G^{*}$ be test design obtained from $G=G_{n,m}$ by deleting all tests of degree larger than $\Gamma$.
If indeed every test of degree at least $\Gamma$ contains at least two infected individuals, then $V_{0+}(G^{*},\CHI)=V_{0+}(G,\CHI)$ and $V_{1+}(G^{*},\CHI)=V_{1+}(G,\CHI)$.
Hence, \Lem~\ref{Prop_removeHighDegs} shows that it suffices to bound $|V_{0+}(G^{*},\CHI)|,|V_{1+}(G^{*},\CHI)|$.
To this end we observe that $G^{*}$ contains few individuals of very high degree.

\begin{lemma} \label{lem_delta_max}
There is $n_0=n_0(\theta,\eps)>0$ such that for all $n>n_0$ and all test designs $G$ with $m\leq \minf$ we have
 $$\abs{\cbc{x\in V_n: \abs{\partial_{G^*} x} > \log^3 n}} \leq\frac{n\log\log n}{\log n}.$$
\end{lemma}
\begin{proof}
Since $\max_{a\in F_m} |\partial_{G^*}a| \leq \Gamma= n^{1-\theta}\log n$, double counting yields
 $$ \sum_{x \in V_n} \abs{\partial_{G^*} x} = \sum_{a \in F_m}|\partial_{G^*}a| \leq \minf\Gamma =O(n\log^2n). $$
Consequently, there are no more than $O(n/\log n)$ individuals $x\in V_n$ with $\abs{\partial_{G^*} x}>\log^3n$.
\end{proof}

\noindent
Further, obtain $G^{(0)}$ from $G^*$ by deleting all individuals of degree greater than $\log^3n$ (but keeping all tests).
Then the degrees of $G^{(0)}$ satisfy
\begin{align}\label{eqG0degs}
\max_{a\in F(G^{(0)})}|\partial_{G^{(0)}}a|&\leq\Gamma,&\max_{x\in V(G^{(0)})}|\partial_{G^{(0)}}x|&\leq\log^3n.
\end{align}
Let $\CHI^{(0)}=(\CHI_x)_{x\in V(G^{(0)})}$ signify the restriction of $\CHI$ to the individuals that remain in $G^{(0)}$.

With these preparations in place we are ready to commence the main step of the proof of \Prop~\ref{Prop_large_theta}.
Given a test design $G$ with $m\leq(1-\eps)\minf$ we are going to construct a sequence $y_1,y_2,\ldots,y_N$, $N=\lceil n^{1-\xi}\rceil$, of individuals of $G^{(0)}$ such that each $y_i$ individually has a moderately high probability of being disguised.
Of course, to conclude that in the end a large number of disguised $y_i$ actually materialise, we need to cope with stochastic dependencies.
To this end we will pick individuals $y_i$ that have pairwise distance at least five in $G^{(0)}$.
The degree bounds \eqref{eqG0degs} guarantee a sufficient supply of such far apart individuals.


To be precise, starting from $G^{(0)}$ we construct a sequence of test designs $G^{(1)},G^{(2)},\ldots,G^{(N)}$ inductively as follows.
For each $i\geq1$ select a variable $y_{i-1}\in V(G^{(i-1)})$ whose probability of being disguised is maximum; ties are broken arbitrarily.
In formulas,
\begin{align*}
\pr\brk{y_{i-1}\in V_+(G^{(i-1)},\CHI^{(i-1)})}&=\max_{y\in V(G^{(i-1)})}\pr\brk{y\in V_+(G^{(i-1)},\CHI^{(i-1)})},
\end{align*}
where, of course, $\CHI^{(i-1)}$ is the only random object.
Then obtain $G^{(i)}$ from $G^{(i-1)}$ by removing $y_{i-1}$ along with all vertices (i.e., tests or individuals) at distance at most four from $y_{i-1}$.
Moreover, let $\CHI^{(i)}$ denote the restriction $(\CHI_x)_{x\in V(G^{(i)})}$ of $\CHI$ to $G^{(i)}$.
The following lemma estimates the probability of $y_i$ being disguised.
 Let $m^*=|F(G^*)|$ be the total number of tests of $G$ of degree at most $\Gamma$.


\begin{lemma}\label{lem_0+_all}
There exists $n_0=n_0(\eps,\theta,\xi)$ such that for all $n>n_0$ and all $G$ with $m\leq(1-\eps)\minf$ we have 
\begin{align*}
\min_{1\leq i\leq N}\pr\brk{y_i\in V_+(G^{(i)})}\geq\exp \bc{-\frac{m\log^22}{n^\theta}-1}.
\end{align*}
\end{lemma}

The proof of \Lem~\ref{lem_0+_all} requires three intermediate steps.
First, we need a lower bound on number of individuals in $G^{(i)}$.
Recall that $N=\lceil n^{1-\xi}\rceil$.

\begin{claim}\label{claim_lem_0+_all_1}
We have $\min_{0\leq i\leq N}|V(G^{(i)})|\geq n- N \Gamma^2\log^6n.$
\end{claim}
\begin{proof}
Since throughout the construction of the $G^{(i)}$ we only delete vertices, the degree bound \eqref{eqG0degs} implies
\begin{align}\label{eqGidegs}
\max_{a\in F(G^{(i)})}|\partial_{G^{(i)}}a|&\leq\Gamma=n^{1-\theta}\log n,&\max_{x\in V(G^{(i)})}|\partial_{G^{(i)}}x|&\leq\log^3n&&\mbox{for all }i\leq N.
\end{align}
We now proceed by induction on $i$.
For $i=0$ there is nothing to show.
Going from $i$ to $i+1\leq N$, we notice that 
because all individuals $x\in V(G^{(i)}) \setminus V(G^{(i+1)})$ have distance at most four from $y_{i+1}$, \eqref{eqGidegs} ensures that
\begin{align}\label{eq_bound_fv2}
|V(G^{(i)}) \setminus V(G^{(i+1)})| &\leq\Gamma^2\log^6n.
\end{align}
Iterating \eqref{eq_bound_fv2}, we obtain 
$|V(G^{(0)}) \setminus V(G^{(i+1)}) | \leq (i+1) \Gamma^2\log^6n$, whence $|V(G^{(i+1)})|\geq n-(i+1)\Gamma^2\log^6n$.
\end{proof}

\noindent
The following claim resembles the proof of~\cite[\Thm~1]{Aldridge_2018} (where the case $k=\Omega(n)$ is considered).

\begin{claim}\label{Claim_FKG}
Let $\cD^{(i)}(x)=\{x\in V_+(G^{(i)})\}$ and let
\begin{align}\label{eqLi}
L^{(i)}&=\frac1{|V(G^{(i)})|}\sum_{x\in V(G^{(i)})}\log\Pr\brk{\cD^{(i)}(x)}.
\end{align}
Then
\begin{align}
L^{(i)}&\geq 
\frac{|F(G^{(i)})|}{|V(G^{(i)})|} \min_{a \in F(G^{(i)})} \abs{\partial_{G^{(i)}} a} \log \bc{1-(1-p)^{|\partial_{G^{(i)}} a|-1}}.\label{eq_disguised4}
\end{align}
\end{claim}
\begin{proof}
For an individual $x\in V(G^{(i)})$ and a test $a\in\partial_{G^{(i)}}x$ let $\cD^{(i)}(x,a)$ be the event that there is another individual $z\in\partial_{G^{(i)}}a\setminus\cbc x$ such that $\CHI_{z}=1$.
Then for every $x\in V(G^{(i)})$ we have
\begin{align}\label{eq_disguised1}
\Pr \brk{\cD^{(i)}\bc x} = \Pr \brk{\bigcap_{a \in \partial_{G^{(i)}} x} \cD^{(i)}(x,a)}.
\end{align}
Furthermore, the events $\cD^{(i)}(x,a)$ are increasing with respect to $\CHI$.
Therefore, \eqref{eq_disguised1} and the FKG inequality imply
\begin{align}\label{eq_disguised2}
\Pr \brk{\cD^{(i)}\bc x}&\geq \prod_{a \in \partial_{G^{(i)}} x} \Pr \brk{\cD^{(i)}(x,a)}.
\end{align}
Moreover, because each entry of $\CHI$ is one with probability $p$ independently, we obtain
\begin{align}\label{eq_disguised2a}
\Pr \brk{\cD^{(i)}(x,a)}=1-(1-p)^{|\partial_{G^{(i)}} a|-1}
\end{align}
Finally, combining \eqref{eq_disguised1}--\eqref{eq_disguised2a}, we obtain
\begin{align*}
     |V(G^{(i)})|L^{(i)} &\geq \sum_{x \in V(G^{(i)})} \sum_{a \in F(G^{(i)})} \vecone \cbc{a \in \partial_{G^{(i)}} x} \log \bc{1-(1-p)^{|\partial_{G^{(i)}} a|-1}} \notag \\
    &=\sum_{a \in F(G^{(i)})} \abs{\partial_{G^{(i)}} a} \log \bc{1-(1-p)^{|\partial_{G^{(i)}} a|-1}} 
    \geq |F(G^{(i)})|\min_{a \in F(G^{(i)})} \abs{\partial_{G^{(i)}}a} \log \bc{1-(1-p)^{|\partial_{G^{(i)}} a|-1}} ,
\end{align*}
as claimed.
\end{proof}

\noindent
As a final preparation for the proof of \Lem~\ref{lem_0+_all} we need the following estimate.

\begin{claim}\label{Claim_z}
The function $z\in(0,\infty)\mapsto z\log(1-(1-p)^{z-1})$ attains its minimum at $z=\bc{1+O(n^{-\Omega(1)})}\log(2)/p$.
\end{claim}
\begin{proof}
We consider three separate cases.
\begin{description}
    \item [Case 1: $z=o(1/p)$] we obtain
    \begin{align}\label{eqClaim_z1}
    z \log \bc{1-(1-p)^{z-1}} &= z \log \bc{1 - \exp \bc{-pz + O(p^2 z)}} = z \log \bc{1 - \bc{1 - pz+ O(p^2z^2)}} \notag\\
    &= \frac z \log (zp + O(zp)^2) = o(1/p).
    \end{align}
    \item [Case 2: $z=\omega(1/p)$] we find
    \begin{align}\label{eqClaim_z2}
    z \log \bc{1-(1-p)^{z-1}} &= z \log \bc{ 1-\exp\bc{-pz + O(p^2 z)}} = -z \bc{\exp(-pz) + O\bc{\exp(-2pz)}}\notag\\
    & = -\frac 1p pz\bc{ \exp \bc{-pz} + \exp \bc{-2pz}} = o(1/p).
    \end{align}
    \item [Case 3: $z=\Theta(1/p)$] letting $d=zp$, we obtain
    \begin{align}\label{eqClaim_z3}
    z \log \bc{1-(1-p)^{z-1}} &= \frac dp \log \bc{1-\exp \bc{-d + O(p)}} = \frac dp \log \bc{1-\exp \bc{-d }} + O(1).
    \end{align}
\end{description}
Since the strictly convex function $d\in(0,\infty)\mapsto d\log(1-\exp(-d))$ attains its minimum at $d=\log 2$,  \eqref{eqClaim_z3} dominates \eqref{eqClaim_z1} and \eqref{eqClaim_z2}.
Thus, the minimiser reads $z=\log(2)/p+O(p^{-1/2})$.
\end{proof}

\begin{proof}[Proof of \Lem~\ref{lem_0+_all}]
Combining Claims~\ref{Claim_FKG} and~\ref{Claim_z}, we see that for all test designs $G$ with $m\leq(1-\eps)\minf$ and for all $i\leq N$,
\begin{align*}
L^{(i)} \geq -\bc{1+O(n^{-\Omega(1)})} \frac{|F(G^{(i)})|\log^22 }{|V(G^{(i)})|p} 
    \geq-\bc{1+O(n^{-\Omega(1)})} \frac{m\log^22 }{|V(G^{(i)})|p}.
\end{align*}
Hence, Claim~\ref{claim_lem_0+_all_1}, 
\eqref{eqxietatheta} and the choice $p=(k+\sqrt k\log n)/n$ imply that for all $i\leq N$,
\begin{align}\label{eq_disguised_case9} 
L^{(i)}\geq-\bc{1+O(n^{-\Omega(1)})} \frac{m\log^22 }{(n-N\Delta^2\log^6n)p}\geq
-\bc{1+O(n^{-\Omega(1)})} \frac{m\log^22 }{n^\theta}.
\end{align}
Further, combining the definition \eqref{eqLi} of $L^{(i)}$ with \eqref{eq_disguised_case9}, we conclude that for every $i\leq N$ there exists an individual $y_i\in V(G^{(i)})$ such that
\begin{align*}
\pr\brk{y_i\in V_+(G^{(i)})}=\Pr\brk{\cD^{(i)}(y_i)}&\geq\exp\bc{L^{(i)}}
    \geq\exp\bc{-\bc{1+O(n^{-\Omega(1)})} \frac{m\log^2(2) }{n^\theta}},
\end{align*}
which implies the assertion.
\end{proof}

\noindent
\Lem~\ref{lem_0+_all} implies the following bound on $|V_{0+}(G^*,\CHI)|$, $|V_{1+}(G^*,\CHI)|$.

\begin{corollary}\label{lotbad}
There exists $n_0=n_0(\eps,\theta,\xi)$ such that for all $n>n_0$ and all $G=G_{n,m}$ with $m\leq(1-\eps)\minf$ we have 
\begin{align*}
\pr\brk{\abs{V_{0+}(G^*,\CHI)}\wedge\abs{V_{1+}(G^*,\CHI)}<\log^4n}<\eps/8.
\end{align*}
\end{corollary}
\begin{proof}
We observe that $V_{+}(G^{(i)},\CHI)\subset V_{+}(G^{*},\CHI)$ for all $i\leq N$ because by construction for any individual $x\in V(G^{(i)})$ every test $a\in\partial_{G^*}x$ of $G^*$ that $x$ belongs to is still present in $G^{(i)}$.
Consequently, we obtain the bound
\begin{align}\label{eqlotbad_1}
\Pr \brk{x\in V_+(G^*)}&\geq \Pr \brk{x\in V(G^{(i)})}&&\mbox{for all $i \in [N]$, $x \in V(G^*)$.} 
\end{align}
Combining \eqref{eqlotbad_1} with \Lem~\ref{lem_0+_all} we obtain  
\begin{align*}
\Pr \brk{y^{(i)}\in V_+(G^*)}&\geq\exp \bc{- \log^2(2) n^{-\theta} m-1}
\geq\exp \bc{- (1-\eps)\log^2(2) n^{-\theta}\minf-1}
&&\mbox{for all $i \in [N]$.} 
\end{align*}
Hence, recalling the definition of $\minf$ from \eqref{eqInfUpper}, we obtain
\begin{align}\label{eqlotbad_2}
\Pr \brk{y^{(i)}\in V_+(G^*)}&\geq\exp \bc{- (1-\eps)\theta\log(n)-1}=n^{(\eps-1)\theta}/\eul.
&&\mbox{for all $i \in [N]$.} 
\end{align}
Since the entry $\CHI_{y^{(i)}}$ is independent of the event $\{y^{(i)}\in V_+(G^*)\}$, the definitions
\eqref{eqdisguise} of $V_{0+}(G^*,\CHI)$ and $V_{1+}(G^*,\CHI)$ and \eqref{eqlotbad_2} yield
\begin{align*}
\Pr \brk{y^{(i)}\in V_{0+}(G^*,\CHI)}&\geq
(1-p)\cdot\frac{n^{(\eps-1)\theta}}\eul\geq\frac{n^{\eps\theta-1}}3,&
\Pr \brk{y^{(i)}\in V_{1+}(G^*,\CHI)}&\geq p\cdot\frac{n^{(\eps-1)\theta}}\eul\geq \frac{n^{\eps\theta-1}}3
&&\mbox{for all $i \in [N]$},
\end{align*}
provided $n$ is sufficiently large.
Therefore, recalling $N=\lceil n^{1-\xi}\rceil$ we obtain for large enough $n$,
\begin{align}\label{eqlotbad_3a}
\Erw\abs{\{y^{(1)},\ldots,y^{(N)}\}\cap V_{0+}(G^*,\CHI)}&\geq n^{\eps\theta-\xi}/3,&
\Erw\abs{\{y^{(1)},\ldots,y^{(N)}\}\cap V_{1+}(G^*,\CHI)}&\geq n^{\eps\theta-\xi}/3.
\end{align}
Further, because the pairwise distances of $y^{(1)},\ldots,y^{(N)}$ in $G^*$ exceed four, the events $\{y^{(i)}\in V_{0+}(G^*,\CHI)\}_{i\leq N}$ are mutually independent.
So are the events $\{y^{(i)}\in V_{1+}(G^*,\CHI)\}_{i\leq N}$.
Finally, since \eqref{eqxietatheta} ensures that $\eps\theta-\xi>0$, \eqref{eqlotbad_3a} and the Chernoff bound yield
\begin{align*}
\pr\brk{\abs{\{y^{(1)},\ldots,y^{(N)}\}\cap V_{0+}(G^*,\CHI)}\leq\log^2n}&\leq\pr\brk{\Bin(N,n^{\eps\theta-1}/3)\leq\log^2n}\leq\exp(-n^{\Omega(1)}),\\
\pr\brk{\abs{\{y^{(1)},\ldots,y^{(N)}\}\cap V_{1+}(G^*,\CHI)}\leq\log^2n}&\leq\pr\brk{\Bin(N,n^{\eps\theta-1}/3)\leq\log^2n}\leq\exp(-n^{\Omega(1)}),
\end{align*}
whence the assertion is immediate.
\end{proof}

\begin{proof}[Proof of \Prop~\ref{Prop_large_theta}]
Suppose that $n>n_0(\eps,\theta,\xi)$ is large enough and let $G=G_{n,m}$ be a test design with $m\leq(1-\eps)\minf$ tests.
If  for every test $a\in F_m$ of degree $|\partial_Ga|>\Gamma$ we have $|\partial_Ga\cap V_1(G,\CHI)|\geq2$, then $V_{0+}(G,\CHI)=V_{0+}(G^*,\CHI)$ and $V_{1+}(G,\CHI)=V_{1+}(G^*,\CHI)$.
Therefore, the assertion is an immediate consequence of \Lem~\ref{Lemma_bin}, \Lem~\ref{Prop_removeHighDegs} and \Cor~\ref{lotbad}.
\end{proof}

\subsection{Proof of \Prop~\ref{Prop_small_theta}}\label{Sec_Prop_small_theta}
Given $\eps>0$ and $\ln(2)/(1+\ln(2))\leq \theta<\theta'<1$ we choose a large enough $n_0=n_0(\eps,\theta,\theta')$ and assume that $n>n_0$.
Furthermore, let $G$ be a test design with $m\leq (1-\eps)\minf(n,\theta)$ for the purpose of identifying $k=\lceil n^\theta\rceil$ infected individuals.
Starting from the test design $G$ infection for density $\theta$ we are going to construct a random test design $\G'$ for infection density $\theta'$ with the same number $m$ of tests as $G$.
The following lemma fixes the order of $\G'$.

\begin{lemma}\label{Lemma_n'}
There exists an integer $n^{\theta/\theta'}/2\leq n'\leq 2n^{\theta/\theta'}\wedge n$ such that $k'=\lceil n'\,^{\theta'}\rceil=k$.
\end{lemma}
\begin{proof}
Let $n''=\lceil n^{\theta/\theta'}/2\rceil$.
Then $(4n'')^{\theta'}>k$ but $n''\,^{\theta'}<k$ because the function $z\in(1,\infty)\mapsto z^{\theta'}$ has derivative less than one.
For the same reason for any integer $n''<N<4n''$ we have $(N+1)^{\theta'}-N^{\theta'}\leq1$ and thus
\begin{align*}
\lceil (N+1)^{\theta'}\rceil-\lceil N^{\theta'}\rceil\leq1.
\end{align*}
Consequently, there exists an integer $n'\in(n'',4n'')$ such that $\lceil n'\,^{\theta'}\rceil=k$.
\end{proof}

Given the test design $G$ with individuals $V_n=\{x_1,\ldots,x_n\}$ and tests $F_m=\{a_1,\ldots,a_m\}$ we now construct the test design $\G'$ as follows.
Choose a subset $V(\G')\subset V_n$ of $n'$ individuals uniformly at random.
Then $\G'$ is the subgraph that $G$ induces on $V(\G')\cup F_m$.
Thus, $\G'$ has the same tests as $G$ but we simply leave out from every test the individuals that do not belong to the random subset $V(\G')$.
Let $\TAU\in\{0,1\}^{V(\G')}$ be a random vector of Hamming weight $k$ and let $\hat\TAU\in\{0,1\}^{F_m}$ be the induced vector of tests results
\begin{align*}
\hat\TAU_a&=\max_{x\in\partial_{\G'}a}\TAU_x&&(a\in F_m).
\end{align*}

\begin{lemma}\label{Lemma_GG'}
For any integer $t>0$ we have $\pr\brk{Z_k(G,\hat\SIGMA)\geq t}\geq\pr\brk{Z_k(\G',\hat\TAU)\geq t}$.
\end{lemma}
\begin{proof}
The choice of $n'$ ensures that $k'=\lceil n'\,^{\theta'}\rceil=k$.
Therefore, the random sets $\{x\in V:\SIGMA_x=1\}$ and $\{x\in V(\G'):\TAU_x=1\}$ are identically distributed.
Indeed, we obtain the latter by first choosing the random subset $V(\G')$ of $V_n$ and then choosing a random subset of $V(\G')$ size $k$.
Clearly, this two-step procedure is equivalent to just choosing a random subset of size $k$ out of $V_n$.
Hence, we can couple $\SIGMA,\TAU$ such that the sets $\{x\in V:\SIGMA_x=1\}$, $\{x\in V:\TAU_x=1\}$ are identical.
Then the construction of $\G'$ ensures that the vectors $\hat\SIGMA$, $\hat\TAU$ coincide as well.

Now consider a vector $\sigma'\in\cS_k(\G',\hat\TAU)$ that explains the test results.
Extend $\sigma'$ to a vector $\sigma\in\{0,1\}^{V_n}$ by setting $\sigma_x=0$ for all $x\in V_n\setminus V(\G')$.
Then $\sigma\in \cS_k(\G,\hat\SIGMA)$.
Hence, $Z_k(G,\hat\SIGMA)\geq Z_k(\G',\hat\TAU)$.
\end{proof}

\begin{proof}[Proof of \Prop~\ref{Prop_small_theta}]
\Lem~\ref{Lemma_GG'} shows that for any $t>0$,
\begin{align*}
\pr\brk{Z_k(G,\hat\SIGMA)\geq t}\geq\pr\brk{Z_k(\G',\hat\TAU)\geq t}
=\Erw\brk{\pr\brk{Z_k(\G',\hat\TAU)\geq t\mid\G'}}.
\end{align*}
Consequently, there exists an outcome $G'$ of $\G'$ such that $\pr\brk{Z_k(G,\hat\SIGMA)\geq t}\geq \pr\brk{Z_k(G',\hat\TAU)\geq t}$.
\end{proof}

\section{The non-adaptive group testing algorithm \SPIV}\label{Sec_alg}

\noindent
In this section we describe the new test design and the associated inference algorithm \SPIV\ for \Thm~\ref{thm_SC}.
Throughout we fix $\theta\in(0,1)$ and $\eps>0$ and we tacitly assume that $n>n_0(\eps,\theta)$ is large enough for the various estimates to hold.

\subsection{The random bipartite graph and the \DD\ algorithm}\label{Sec_random_bip}
To motivate the new test design we begin with a brief discussion of the plain random design used in prior work and the best previously known inference algorithm \DD~\cite{Coja_2019,Johnson_2019}.
At first glance a promising candidate test design appears to be a random bipartite graph with one vertex class $V_n=\{x_1,\ldots,x_n\}$ representing individuals and the other class $F_m=\{a_1,\ldots,a_m\}$ representing tests.
Indeed, two slightly different random graph models have been proposed~\cite{Aldridge_2019}.
First, in the {\em Bernoulli model} each $V_n$--$F_m$-edge is present with a certain probability (the same for every pair) independently of all others.
However, due to the relatively heavy lower tail of the degrees of the individuals, this test design turns out to be inferior to a second model where the degrees of the individuals are fixed.
Specifically, in the {\em $\Delta$-out model} every individual independently joins an equal number of $\Delta$ tests drawn uniformly at random without replacement~\cite{Mezard_2008}.

Clearly, in order to extract the maximum amount of information $\Delta$ should be chosen so as to maximise the entropy of the vector of test results.
Specifically, since the average test degree equals $\Delta n/m$ and a total of $k$ individuals are infected, the average number of infected individuals per test comes to $\Delta k/m$.
Indeed, since $k\sim n^\theta$ for a fixed $\theta<1$, the number of infected individuals in test $a_i$ can be well approximated by a Poisson variable.
Therefore, setting
\begin{align}\label{eqDelta_1}
\Delta&\sim\frac mk\ln2
\end{align}
ensures that about half the tests are positive  \whp{}

With respect to the performance of the $\Delta$-out model, \cite[\Thm~1.1]{Coja_2019} implies together with \Thm~\ref{opt} that this simple construction is information-theoretically optimal.
Indeed, $m=(1+\eps+o(1))\minf$ test suffice so that an exponential time algorithm correctly infers the set of infected individuals.
Specifically, the algorithm solves a minimum hypergraph vertex cover problem with the individuals as the vertex set and the positive test groups as the hyperedges.
For $m=(1+\eps+o(1))\minf$ the unique optimal solution is precisely the correct set of infected individuals \whp{}
While the worst case NP-hardness of hypergraph vertex cover does not, of course, preclude the existence of an algorithm that is efficient on random hypergraphs, despite considerable efforts no such algorithm has been found.
In fact, as we saw in \Sec~\ref{Sec_related} for a good number of broadly similar inference and optimisation problems on random graphs no efficient information-theoretically optimal algorithms are known.

But for $m$ exceeding the threshold $\malg$ from \eqref{eqDD} an efficient greedy algorithm {\tt DD} correctly recovers $\SIGMA$ \whp{}
The algorithm proceeds in three steps.
\begin{description}
\item[DD1] declare every individual that appears in a negative test uninfected and subsequently remove all negative tests and all individuals that they contain.
\item[DD2] for every remaining (positive) test of degree one declare the individual that appears in the test infected.
\item[DD3] declare all other individuals as uninfected.
\end{description}
The decisions made by the first two steps {\bf DD1}--{\bf DD2} are clearly correct but {\bf DD3} might produce false negatives.
Prior to the present work {\tt DD} was the best known polynomial time group testing algorithm.
While {\tt DD} correctly identifies the set of infected individuals \whp\ if $m>(1+\eps)\malg$~\cite{Johnson_2019},
%
%
the algorithm  fails if $m<(1-\eps)\malg$ \whp~\cite{Coja_2019}.

\subsection{Spatial coupling}\label{Sec_test_design}
The new efficient algorithm \SPIV\ for \Thm~\ref{thm_SC} that gets by with the optimal number $(1+\eps+o(1))\minf$ of tests comes with a tailor-made test design that, inspired by spatially coupled codes~\cite{Felstrom_1999, Kudekar_2011, Kudekar_2013}, combines randomisation with a superimposed geometric structure.
Specifically, we divide both the individuals and the tests into 
\begin{align}\label{eqell}
\ell=\lceil \ln^{1/2}n\rceil
\end{align}
compartments of equal size.
The compartments are arranged along a ring and each individual joins an equal number of random tests in the
\begin{align}\label{eqs}
s=\lceil\ln\ln n\rceil=o(\ell)
\end{align}
topologically subsequent compartments.
Additionally, to get the algorithm started we equip the first $s$ compartments with extra tests so that they can be easily diagnosed via the {\tt DD} algorithm.
Then, having diagnosed the initial compartments correctly, \SPIV\ will work its way along the ring, diagnosing one compartment after the other.

\FigSP

To implement this idea precisely 
we partition the set $V=V_n=\{x_1,\ldots,x_n\}$ of individuals into pairwise disjoint subsets $V[1],\ldots,V[\ell]$ of sizes $|V[j]|\in\{\lfloor n/\ell\rfloor,\lceil n/\ell\rceil\}$.
With each compartment $V[i]$ of individuals we associate a compartment $F[i]$ of tests of size $\abs{F[i]}=m/\ell$ for an integer $m$ that is divisible by $\ell$.
Additionally, we introduce a set $F[0]$ of $10\lceil (k s/\ell)\log n\rceil$ extra tests to facilitate the greedy algorithm for diagnosing the first $s$ compartments.
Thus, the total number of tests comes to
\begin{align}\label{eqtotalTests}
|F[0]|+\sum_{i=1}^\ell|F[i]|=(1+O(s/\ell))m=(1+o(1))m.
\end{align}
Finally, for notational convenience we define $V[\ell+i]=V[i]$ and $F[\ell+i]=F[i]$ for $i=1,\ldots,s$.

The test groups are composed as follows:
let
\begin{align}\label{eqDelta}
k&=\lceil n^\theta\rceil&\mbox{and let}&&\Delta&=\frac{m\log2}k+O(s)
\end{align}
be an integer divisible by $s$; cf.~\eqref{eqDelta_1}.
Then we construct a random bipartite graph as follows.
\begin{description}
\item[SC1] for $i=1,\ldots,\ell$ and $j=1,\ldots,s$ every individual $x\in V[i]$ joins $\Delta/s$ tests from $F[i+j-1]$ 
chosen uniformly at random without replacement. The choices  are mutually independent for all individuals $x$ and all $j$.
\item[SC2] additionally, each individual from  $V[1]\cup\cdots\cup V[s]$ independently joins $\lceil 10 \log(2)\log n \rceil$ random tests from $F[0]$, drawn uniformly without replacement.
\end{description}
Thus, {\bf SC1} provides that the individuals in compartment $V[i]$ take part in the next $s$ compartments $F[i],\ldots,F[i+s-1]$ of tests along the ring.
Furthermore, {\bf SC2} supplies the tests required by the {\tt DD} algorithm to diagnose the first $s$ compartments.
Figure~\ref{Fig_spatial_coupling_idea} provides an illustration of the resulting random test design,

From here on the test design produced by {\bf SC1--SC2} is denoted by $\G$.
Furthermore $\SIGMA\in\{0,1\}^V$ denotes a uniformly random vector of Hamming weight $k$, drawn independently of $\G$,  and  $\hat\SIGMA=(\hat\SIGMA_a)_{a\in F[0]\cup\cdots\cup F[\ell]}$ signifies the vector of test results
\begin{align*}
\hat\SIGMA_a=\max_{x\in\partial a}\SIGMA_x.
\end{align*}
In addition, let $V_1=\{x\in V:\SIGMA_x=1\}$ be the set of infected individuals and let $V_0=V\setminus V_1$ be the set of healthy individuals.
Moreover, let $F=F[0]\cup F[1]\cup\cdots\cup F[\ell]$ be the set of all tests, let $F_1=\{a\in F:\hat\SIGMA_a=1\}$ be the set of all positive tests and let $F_0=F\setminus F_1$ be the set of all negative tests.
Finally, let
\begin{align*}
V_0[i]&=V[i]\cap V_0,&V_1[i]&=V[i]\cap V_1,&F_0[i]&=F[i]\cap F_0,&F_1[i]&=F[i]\cap F_1.
\end{align*}

The following proposition summarises a few basic properties of the test design $\G$.

\begin{proposition}\label{Prop_basic}
If $m=\Theta(n^\theta\log n)$ then $\G$ enjoys the following properties with probability $1-o(n^{-2})$.
\begin{enumerate}[(i)]
\item  The infected individual counts in the various compartments satisfy
\begin{align*}
\frac k\ell-\sqrt{\frac k\ell}\log n\leq\min_{i\in[\ell]}|V_1[i]|\leq\max_{i\in[\ell]}|V_1[i]|\leq \frac k\ell+\sqrt{\frac k\ell}\log n.
\end{align*}
\item For all $i\in[\ell]$ and all $j\in[s]$ the test degrees satisfy
\begin{align*}
\frac{\Delta n}{ms}-\sqrt{\frac{\Delta n}{ms}}\log n &\leq \min_{a\in F[i+j-1]}|V[i]\cap\partial a| \leq \max_{a\in F[i+j-1]}|V[i]\cap\partial a| \leq \frac{\Delta n}{ms} + \sqrt{\frac{\Delta n}{ms}}\log n.
\end{align*}
\item For all $i\in[\ell]$ the number of negative tests in compartment $F[i]$ satisfies
\begin{align*}
\frac{m}{2\ell}-\sqrt m\ln^3 n\leq\abs{F_0[i]}\leq\frac m{2\ell}+\sqrt m\ln^3 n.
\end{align*}
\end{enumerate}
\end{proposition}

\noindent
We prove \Prop~\ref{Prop_basic} in \Sec~\ref{Sec_Prop_basic}.
Finally, as a preparation for things to come we point out that for any specific individual $x\in V[i]$ and any particular test $a\in F[i+j]$, $j=0,\ldots,s-1$,  we have
\begin{align}\label{eqax}
\pr\brk{x\in\partial a}&=1-\pr\brk{x\not\in\partial a}=1-\binom{|F[i+j]|-1}{\Delta/s}\binom{|F[i+j]|}{\Delta/s}^{-1}
=\frac{\Delta\ell}{ms}+O\bc{\bcfr{\Delta\ell}{ms}^2}.
\end{align}

\subsection{The Spatial Inference Vertex Cover (`\SPIV') algorithm}
The \SPIV\ algorithm for \Thm~\ref{thm_SC} proceeds in three phases.
The plan of attack is for the algorithm to work its way along the ring, diagnosing one compartment after the other aided by what has been learned about the preceding compartments.
Of course, we need to start somewhere.
Hence, in its first phase \SPIV\ diagnoses the seed compartments $V[1],\ldots,V[s]$.

\subsubsection{Phase 1: the seed.}
Specifically, the first phase of \SPIV\ applies the {\tt DD} greedy algorithm from \Sec~\ref{Sec_random_bip} to the subgraph of $\G$ induced on the individuals $V[1]\cup\ldots\cup V[s]$ and the tests $F[0]$.
Throughout the vector $\tau\in\{0,1\}^{V}$ signifies the algorithm's current estimate of the ground truth $\SIGMA$.

\IncMargin{1em}
\begin{algorithm}[ht]
 \KwData{$\G$, $\hat\SIGMA$}
 \KwResult{an estimate of $\SIGMA$}
   Let $(\tau_x)_{x\in V[1]\cup\cdots\cup V[s]}\in\cbc{0,1}^{V[1]\cup\cdots\cup V[s]}$ be the result of applying {\tt DD} to the tests $F[0]$\;
   Set $\tau_{x}=0$ for all individuals $x\in V\setminus(V[1]\cup\cdots\cup V[s])$\;
 \caption{\SPIV, phase~1}\label{SC_algorithm1}
 \label{Alg_SC}
\end{algorithm}
\DecMargin{1em}

\noindent
The following proposition, whose proof can be found in \Sec~\ref{Sec_prop_seed}, summarises the analysis of  phase~1.

\begin{proposition} \label{prop_seed}
\Whp\ the output of~\DD\ satisfies  $\tau_x=\SIGMA_x$ for all $x\in V[1]\cup\cdots\cup V[s]$.
\end{proposition}


\subsubsection{Phase 2: enter the ring.}
This is the main phase of the algorithm.
Thanks to \Prop~\ref{prop_seed} we may assume that the seed has been diagnosed correctly.
Now, the programme is to diagnose one compartment after the other, based on what the algorithm learned previously.
Hence, assume that we managed to diagnose compartments $V[1],\ldots,V[i]$ correctly.
How do we proceed to compartment $V[i+1]$?

For a start, we can safely mark as uninfected all individuals in $V[i+1]$ that appear in a negative test.
But a simple calculation reveals that this will still leave us with many more than $k$
undiagnosed individuals \whp\ 
To be precise, consider the set of uninfected disguised individuals
\begin{align*}
\zeroplus[i+1]&=\cbc{x\in V_0[i+1]:\hat\SIGMA_a=1\mbox{ for all }a\in\partial x},
\end{align*}
i.e., uninfected individuals that fail to appear in a negative test.
In \Sec~\ref{Sec_lemma_v0+} we prove the following.

\begin{lemma}\label{lemma_v0+}
{Suppose that $(1+\eps)\madapt\leq m=O(n^\theta\log n)$.}
Then \whp\ for all $s\leq i<\ell$ we have  $$\abs{\zeroplus[i+1]}=\bc{1+O \bc{n^{-\Omega(1)}}}\frac n{\ell2^{\Delta} }.$$
\end{lemma}			

\noindent
Hence, by the definition \eqref{eqDelta} of $\Delta$ for $m$ close to $\minf$ the set $\zeroplus[i+1]$ has size $k^{1+\Omega(1)}\gg k$ \whp{}

Thus, the challenge is to discriminate between $\zeroplus[i+1]$ and the set $V_1[i+1]$ of actual infected individuals in compartment $i+1$.
 The key observation is that we can tell these sets apart by counting currently `unexplained' positive tests. 
To be precise, for an individual $x\in V[i+1]$ and $1\leq j\leq s$ let $ \vW_{x,j}$ be the number of tests in compartment $F[i+j]$ that contain $x$ but that do not contain an infected individual from the preceding compartments $V[1]\cup\cdots\cup V[i]$.
In formulas,
\begin{align}\label{eqWxj}
\vW_{x,j} = \abs{\cbc{a\in\partial x \cap F[i+j]:
 		\partial a\cap(V_1[1]\cup\cdots\cup V_1[i])=\emptyset}}.
\end{align}
Crucially, the following back-of-the-envelope calculation shows that the mean of this random variable depends on whether $x$ is infected or healthy but disguised.
\begin{description}
\item[Infected individuals ({$x\in V_1[i+1]$})] 
consider a test $a\in \partial x\cap F[i+j]$, $j=1,\ldots,s$.
Because the individuals join tests independently, conditioning on $x$ being infected does not skew the distribution of the individuals from the $s-j$ prior compartments $V[i+j-s+1],\ldots,V[i]$ that appear in $a$.
Furthermore, we chose $\Delta$ so that for each of these compartments $V[h]$  the expected number of infected individuals that join $a$ 
has mean $(\log2)/s$.
Indeed, due to independence it is not difficult to see that $|V_1[h]\cap\partial a|$ is approximately a Poisson variable.
Consequently,
\begin{align}\label{eqAmin1}
\pr\brk{(V_1[i+j-s+1]\cup\cdots\cup V_1[i])\cap\partial a=\emptyset}\sim2^{-(s-j)/s}.
\end{align}
Hence, because $x$ appears in $\Delta/s$ tests $a\in F[i+j]$, the linearity of expectation yields
\begin{align}\label{eqAmin2}
\Erw\brk{\vW_{x,j}\mid x\in V_1[i+1]}\sim 2^{j/s-1}\frac{\Delta}s.
\end{align}
\item[Disguised healthy individuals ({$x\in V_{0+}[i+1]$})]
similarly as above, for any individual $x\in V[i+1]$ and any $a\in\partial x\cap F[i+j]$ the {\em unconditional} number of infected individuals in  $a$ is asymptotically $\Po(\ln 2)$.
But given $x \in V_{0+}[i+1]$ we know that $a$ is positive.
Thus, $\partial a\setminus\cbc x$ contains at least one infected individual.
In effect, the number of positives in $a$ approximately turns into a conditional Poisson $\Po_{\geq1}(\ln 2)$.
Consequently, for test $a$ not to include any infected individual from one of the known compartments $V[h]$, $h=i+j-s+1,\ldots,i$,
every infected individual in test $a$ must stem from the $j$ yet undiagnosed compartments. 
Summing up the conditional Poisson and recalling that $x$ appears in $\Delta/s$ tests $a\in F[j]$, we thus obtain
\begin{align}\label{eqAmin3}
\Erw\brk{\vW_{x,j}\mid x\in V_{0+}[i+1]}\sim \frac \Delta s \sum_{t \geq 1} \Pr \brk{\Po_{\geq 1}(\ln 2) = t}(j/s)^t = (2^{j/s}-1)\frac\Delta s.
\end{align}
 \end{description}

A first idea to tell $V_{0+}[i+1]$ and $V_1[i+1]$ apart might thus be to simply calculate
\begin{align}\label{equnweighted}
\vW_{x}&=\sum_{j=1}^{s-1}\vW_{x,j}&&(x\in V[i+1]).
\end{align}
Indeed, \eqref{eqAmin2} and \eqref{eqAmin3} yield
\begin{align*}
\Erw\brk{\vW_x\mid x\in V_{1}[i+1]}&\sim\frac\Delta{2\log2}= 0.721\ldots\,\Delta&\mbox{whereas}&&
\Erw\brk{\vW_x\mid x\in V_{0+}[i+1]}&\sim\frac{\Delta(1-\ln 2)}{\ln2}= 0.442\ldots\,\Delta.
\end{align*}
But unfortunately a careful large deviations analysis reveals that $\vW_{x}$ is not sufficiently concentrated.
More precisely, even for $m=(1+\eps +o(1))\minf$ there are as many as $k^{1+\Omega(1)}$ `outliers' $x\in V_{0+}[i+1]$ whose $\vW_x$ grows as large as the mean $\Delta/(2\log2)$ of actual infected individuals \whp

At second thought the plain sum \eqref{equnweighted}  does seem to leave something on the table.
While $\vW_x$ counts all as yet unexplained positive tests equally,  not all of these tests reveal the same amount of information.
In fact, we should really be paying more attention to `early' unexplained tests $a\in F[i+1]$ than to `late' ones $b\in F[i+s]$.
For we already diagnosed $s-1$ out of the $s$ compartments of individuals that $a$ draws on, whereas only one of the $s$ compartments that contribute to $b$ has already been diagnosed.
Thus, the unexplained test $a$ is a much stronger indication that $x$ might be infected.
Consequently, it seems promising to replace $\vW_x$ by a weighted sum
\begin{align}\label{eqWstar}
\vW_x^\star&=\sum_{j=1}^{s-1}w_j\vW_{x,j}
\end{align}
with $w_1,\ldots,w_{s-1}\geq0$ chosen so as to gauge the amount of information carried by the different compartments.

To find the optimal weights $w_1,\ldots,w_{s-1}$ we need to investigate the rate function of $\vW_x^\star$ given $x\in V_{0+}[i+1]$.
More specifically, we should minimise the probability that $\vW_x^\star$ given $x\in V_{0+}[i+1]$ grows as large as the mean of  $\vW_x^\star$ given $x\in V_{1}[i+1]$, which we read off \eqref{eqAmin2} easily:
\begin{align}\label{eqAmin4}
\Erw \brk{\vW^\star_x \mid x \in V_1[i+1]}&\sim\frac{\Delta}{s} \wq.
\end{align}
A careful large deviations analysis followed by a Lagrangian optimisation leads to the optimal choice
\begin{align}\label{eqwj}
w_j&=\ln \frac{(1 - 2 \zeta) 2^{j/s-1} (2 - 2^{j/s})}{(1 - (1 - 2\zeta)2^{j/s-1}) (2^{j/s}-1)}
&&\mbox{where}&&\zeta=1/s^2.
\end{align}

The following two lemmas show that with these weights the scores $\vW_x^\star$ discriminate well between the potential false positives and the infected individuals.
More precisely, thresholding $\vW_x^\star$ we end up misclassifying no more than $o(k)$ individuals $x$ \whp{}

\begin{lemma}\label{lem_1dev}
Suppose that $(1+\eps)\madapt\leq m=O(n^\theta\log n)$. \Whp~we have
\begin{align}\label{eqlem_1dev}
\sum_{s\leq i<\ell}\sum_{x\in \one[i]}\vecone\cbc{\wxstar x < (1-\zeta/2) \frac \Delta s \wq}\leq k\exp\bc{-\frac{\Omega(\log n)}{(\log\log n)^4}}.
\end{align}
\end{lemma}

\begin{lemma}\label{lem_large_deviation0+}
Suppose that $(1+\eps)\madapt\leq m=O(n^\theta\log n)$. \Whp~we have
\begin{align}\label{eqlem_0+dev}
 \sum_{s\leq i<\ell}\sum_{x\in\zeroplus[i]}\vecone\cbc{\wxstar x > (1-2\zeta) \frac \Delta s \wq}&\leq k^{1-\Omega(1)}.
\end{align}
\end{lemma}

\noindent
We prove these two lemmas in \Sec s~\ref{Sec_lem_1dev} and \ref{Sec_lem_large_deviation0+}.


\Lem s~\ref{lem_1dev}--\ref{lem_large_deviation0+} leave us with only one loose end.
Namely, calculating the scores $\vW_x^\star$ requires knowledge of the correct infection status $\SIGMA_x$ of {\em all} the individuals $x\in V[1]\cup\cdots\cup V[i]$ from the previous compartments.
But since the r.h.s.\ expressions in \eqref{eqlem_1dev} and \eqref{eqlem_0+dev} are non-zero, it is unrealistic to assume that the algorithm's estimates $\tau_x$ will consistently match the ground truth $\SIGMA_x$ beyond the seed compartments.
Hoping that the algorithm's estimate will not stray too far, we thus have to  make do with the  approximate scores
\begin{align}\label{eqWtau}
W_x^\star(\tau)& = \sum_{j=1}^{s-1} w_j W_{x,j}(\tau),&\mbox{where}&&
W_{x,j}(\tau)&=\abs{\cbc{a\in\partial x \cap F[i+j-1]: \max_{y \in \partial a\cap(V[1]\cup\cdots V[i])}\tau_y = 0}}.
\end{align}
Hence, phase~2 of \SPIV\ reads as follows.

\IncMargin{1em}
\begin{algorithm}[ht]
  \setcounter{AlgoLine}{2}
  \For{$i=s, \dots, \ell-1 $}{
 \For{$x \in V[i+1]$}{
  \If{$\exists a \in \partial x: \hat\SIGMA_a=0$}{$\tau_x = 0$ \tcp*[h]{classify as uninfected}}
  \ElseIf{$W_x^\star(\tau) < (1-\zeta) \frac \Delta s \wq$}{
 $\tau_x = 0$ \tcp*[h]{tentatively classify as uninfected}}
 \Else{$\tau_x = 1$ \tcp*[h]{tentatively classify as infected}}
  }}
\caption{\SPIV, phase~2.}
 \end{algorithm}
\DecMargin{1em}

Since phase~2 of \SPIV\ uses the approximations from \eqref{eqWtau}, there seems to be a risk of errors amplifying as we move along.
Fortunately,  it turns out that errors proliferate only moderately and the second phase of \SPIV\ will misclassify only $o(k)$ individuals.
The following proposition summarises the analysis of phase~2.

\begin{proposition} \label{prop_dist_psi}
Suppose that $(1+\eps)\madapt\leq m=O(k\log n)$.
\Whp\ the assignment $\tau$ obtained after steps 1--10 satisfies $$\sum_{x\in V}\vecone\cbc{\tau_x\neq\SIGMA_x}\leq k\exp\bc{-\frac{\log n}{(\ln\ln n)^6}}.$$
\end{proposition}

\noindent
The proof of \Prop~\ref{prop_dist_psi} can be found in \Sec~\ref{sec_prop_dist_psi}.

\subsubsection{Phase 3: cleaning up.}
The final phase of the algorithm rectifies the errors incurred during phase~2.
The combinatorial insight that makes this possible is that for $m\geq(1+\eps)\minf$ every infected individual has at least $\Omega(\Delta)$ positive tests to itself \whp{}
Thus, these tests do not feature a second infected individual.
Phase~3 of the algorithm exploits this observation by simply thresholding the number $S_x$ of tests where there is no other infected individual besides potentially $x$.
Thanks to the expansion properties of the graph $\G$, each iteration of the thresholding procedure reduces the number of misclassified individuals by at least a factor of three.
In effect, after $\log n$ iterations all individuals will be classified correctly \whp{}
Of course, due to \Prop~\ref{prop_seed} we do not need to reconsider the seed  $V[1]\cup\cdots\cup V[s]$.

\IncMargin{1em}
\begin{algorithm}[ht]
  \setcounter{AlgoLine}{10}
  Let $\tau^{(1)}=\tau$\;
  \For{ $i = 1,\dots, \lceil\log n\rceil$} {
  For all $x\in V[s+1]\cup\cdots\cup V[\ell]$ calculate\\
   $\displaystyle \qquad S_x(\tau^{(i)})=\sum_{a\in\partial x:\hat\SIGMA_a=1}\vecone\cbc{\forall y\in\partial a\setminus\cbc x:\tau^{(i)}_y=0}$\;
  Let 
  $\displaystyle\tau_x^{(i+1)}=
  	\begin{cases}
  	\tau_x^{(i)}&\mbox{ if }x\in V[1]\cup\cdots\cup V[s],\\
  	\vecone\cbc{S_x \bc{\tau^{(i)}}>\ln^{1/4}n}&\mbox{ otherwise }
  	\end{cases}$\;
  	}
  \KwRet{$\tau^{(\lceil\log n\rceil)}$}
\caption{\SPIV, phase~3.}\label{SC_algorithm}
\end{algorithm}
\DecMargin{1em}

\begin{proposition}\label{prop_endgame}
Suppose that $(1+\eps)\minf\leq m=O(n^\theta\log n)$.
\Whp\ for all $1\leq i\leq\lceil\log n\rceil$ we have
\begin{align*}
\sum_{x\in V}\vecone\{\tau^{(i+1)}_x\neq \SIGMA_x\}\leq\frac13\sum_{x\in V}\vecone\{\tau^{(i)}_x\neq \SIGMA_x\}.
\end{align*}
\end{proposition}

\noindent
We prove \Prop~\ref{prop_endgame} in \Sec~\ref{sec_prop_endgame}.

\begin{proof}[Proof of \Thm~\ref{thm_SC}]
The theorem is an immediate consequence of \Prop s~\ref{prop_seed}, \ref{prop_dist_psi} and~\ref{prop_endgame}.
\end{proof}

\subsection{Proof of \Prop~\ref{Prop_basic}}\label{Sec_Prop_basic}

The number $|V_1[i]|$ of infected individuals in compartment $V[i]$ has distribution $\Hyp(n,k,|V[i]|)$.
Since $||V[i]|-n/\ell|\leq1$, (i) is an immediate consequence of the Chernoff bound from \Lem~\ref{lem_hyperchernoff}.

With respect to (ii), we recall from \eqref{eqax} that 
$\pr\brk{x\in\partial a}=\frac{\Delta\ell}{ms}(1+O(\frac{\Delta\ell}{ms}))$.
Hence, because the various individuals $x\in V[i]$ join tests independently, the number $|V[i]\cap\partial a|$ of test participants from $V[i]$ has distribution
\begin{align*}
|V[i]\cap\partial a|\sim\Bin(|V[i]|,\Delta\ell/\bc{ms}+O((\Delta\ell/{ms})^2)).
\end{align*}
Since $|V[i]|=n/\ell+O(1)$, assertion (ii) follows from \eqref{eqDelta} and the Chernoff bound from \Lem~\ref{lem_chernoff}.

Coming to (iii), due to part (i) we may condition on $\cE=\{\forall i\in[\ell]:|V_1[i]|=k/\ell+O(\sqrt{k/\ell}\log n)\}.$
Hence,  with $h$ ranging over the $s$ compartments whose individuals join tests in $F[i]$, \eqref{eqax} implies that for every test $a\in F[i]$ the number of infected individuals $\abs{V_1\cap\partial a}$ is distributed as a sum of independent binomial variables
\begin{align*}
\abs{V_1\cap\partial a}&\sim\sum_{h}\vX_h&&\mbox{with}&\vX_h\sim\Bin\bc{V_1[h],\frac{\Delta\ell}{ms}+O\bc{\bcfr{\Delta\ell}{ms}^2}}.
\end{align*}
Consequently, \eqref{eqDelta} ensures that the event $V_1\cap\partial a=\emptyset$ has conditional probability
\begin{align*}
\pr\brk{V_1\cap\partial a=\emptyset\mid\cE}&=
\prod_h\pr\brk{\vX_h=0\mid\cE}=\exp\brk{s\bc{\frac k\ell+O\bc{\sqrt{\frac k\ell}\log n}}\log\bc{1-\frac{\Delta\ell}{ms}+O\bc{\bcfr{\Delta\ell}{ms}^2}}}\\
&=\exp\brk{-\frac{sk}{\ell}\cdot\frac{\Delta\ell}{ms}+O\bc{\sqrt{\frac{k}\ell}\cdot\frac{\Delta\ell}{m}}+
	O\bc{\frac{sk}{\ell}\cdot \bcfr{\Delta\ell}{ms}^2}}=\frac12+O(\sqrt{\ell/k}).
\end{align*}
Therefore, we obtain the estimate
\begin{align}\label{eqneg1}
\Erw\brk{|F_0[i]| \mid\cE}&=\frac{m}{2\ell}+O(\sqrt m\log n).
\end{align}
Finally, changing the set of tests that a specific infected individual $x\in V_1[h]$ joins shifts $|F_0[i]|$ by at most $\Delta$
(while tinkering with uninfected ones does not change $|F_0[i]|$ at all).
Therefore, the Azuma--Hoeffding inequality yields
\begin{align}\label{eqneg2}
\pr\brk{\abs{|F_0[i]|-\Erw\brk{|F_0[i]| \mid\cE}}\geq t\mid\cE}&\leq2\exp\bc{-\frac{t^2}{2k\Delta^2}}&&\mbox{for any }t>0.
\end{align}
Thus, (iii) follows from \eqref{eqDelta}, \eqref{eqneg1} and \eqref{eqneg2} with $t=\sqrt m\log^3n$.

\subsection{Proof of \Prop~\ref{prop_seed}}\label{Sec_prop_seed}
Let $D=\lceil 10\log(2)\log n\rceil$ and recall that $|F[0]|=\lceil 10ks\log(n)/\ell\rceil$.
Since by {\bf SC2} every individual from $\in V[1]\cup\cdots\cup V[s]$ joins $D$ random tests from $F[0]$,
in analogy to \eqref{eqax} for every $x\in V[1]\cup\cdots\cup V[s]$ and every test $a\in F[0]$ we obtain
\begin{align}\label{eqax0}
\pr\brk{x\in\partial a}&=1-\pr\brk{x\not\in\partial a}=1-\binom{|F[0]|-1}{D}\binom{|F[0]|}{D}^{-1}
=\frac{D}{|F[0]|}\bc{1+O\bcfr{D}{|F[0]|}}
=\frac{\ell\log2}{ks}\bc{1+O(n^{-\Omega(1)})}.
\end{align}
Let $F_1[0]$ be the set of tests $a\in F[0]$ with $\hat\SIGMA_a=1$.

\begin{lemma}\label{Lemma_F0pos}
\Whp\ the number of positive tests $a\in F[0]$ satisfies $|F_1[0]|=|F[0]|(\frac12+O(n^{-\Omega(1)}))$.
\end{lemma}
\begin{proof}
By \Prop~\ref{Prop_basic} we may condition on the event $\cE$ that $|V_1[1]\cup\cdots\cup V_1[s]|=\frac{ks}\ell(1+O(n^{-\Omega(1)}))$.
Hence, \eqref{eqax0} implies that given $\cE$ the expected number of infected individuals in a test $a\in F[0]$ comes to
\begin{align}\label{eqLemma_F0pos1}
\Erw[|\partial a\cap V_1|\mid\cE]=\log 2+O(n^{-\Omega(1)}).
\end{align}
Moreover, since individuals join tests independently, $|\partial a\cap V_1|$ is a binomial random variable.
Hence, \eqref{eqLemma_F0pos1} implies $\pr[\partial a\cap V_1=\emptyset\mid\cE]=\frac12+O(n^{-\Omega(1)})$.
Consequently, since $\pr\brk\cE=1-o(n^{-2})$ by \Prop~\ref{Prop_basic},
\begin{align}\label{eqLemma_F0pos2}
\Erw|F_1\cap F[0]|=\Erw|F_1[0]|&=\frac{|F[0]|}2(1+O(n^{-\Omega(1)})).
\end{align}
Finally, changing the set $\partial x$ of neighbours of an infected individual can shift $|F_1[0]|$ by at most $D$.
Therefore, the Azuma--Hoeffding inequality implies that
\begin{align}\label{eqLemma_F0pos3}
\pr\brk{||F_1[0]|-\Erw|F_1[0]||>t}&\leq2\exp\bc{-\frac{t^2}{2D^2k}}&&\mbox{for any }t>0.
\end{align}
Since $D=O(\log n)$, combining \eqref{eqLemma_F0pos2} and \eqref{eqLemma_F0pos3} and setting, say, $t=k^{2/3}$ completes the proof.
\end{proof}

As an application of \Lem~\ref{Lemma_F0pos} we show that \whp\ every seed individual $x$ appears in a test $a\in F[0]$ whose other individuals are all healthy.

\begin{corollary}\label{Cor_F0pos}
\Whp\ every individual $x\in V[1]\cup\cdots\cup V[s]$ appears in a test $a\in F[0]\cap\partial x$ such that $\partial a\setminus\cbc x\subset V_0$.
\end{corollary}
\begin{proof}
We expose the random bipartite graph induced on $V[1]\cup\cdots\cup V[s]$ and $F[0]$ in two rounds.
In the first round we expose $\SIGMA$ and all neighbourhoods $(\partial y)_{y\in (V[1]\cup\cdots\cup V[s])\setminus\cbc x}$.
In the second round we expose $\partial x$.
Let $\vX$ be the number of negative tests $a\in F[0]$ after the first round.
Since $x$ has degree $D=O(\log n)$, \Lem~\ref{Lemma_F0pos} implies that $\vX=|F[0]|(\frac12+O(n^{-\Omega(1)}))$ \whp\
Furthermore, given $\vX$ the number of tests $a\in\partial x$ all of whose other individuals are uninfected has distribution $\Hyp(|F[0]|,\vX,D)$.
Hence,
\begin{align}\label{eqCor_F0pos}
\pr\brk{\forall a\in\partial x:V_1\cap\partial a\setminus\cbc x\neq\emptyset\mid\vX}&=\binom{|F[0]|-\vX}{D}\binom{|F[0]|}{D}^{-1}\leq\exp(-D\vX/|F[0]|).
\end{align}
Assuming $\vX/|F[0]|=\frac12+O(n^{-\Omega(1)})$ and recalling that $D=\lceil 10\log(2)\log n\rceil$, we obtain  $\exp(-D\vX/|F[0]|)=o(1/n)$.
Thus, the assertion follows from \eqref{eqCor_F0pos} and the union bound.
\end{proof}

\begin{proof}[Proof of \Prop~\ref{prop_seed}]
Due to \Cor~\ref{Cor_F0pos} we may assume that for every $x\in V[1]\cup\cdots\cup V[s]$ there is a test $a_x\in F[0]$ such that $\partial a_x\setminus\cbc x\subset V_0$.
Hence, recalling the \DD\ algorithm from \Sec~\ref{Sec_random_bip}, we see that the first step {\bf DD1} will correctly identify all healthy individuals $x\in V_0[1]\cup\cdots\cup V_0[s]$.
Moreover, the second step {\bf DD2} will correctly classify all remaining individuals $V_1[1]\cup\cdots\cup V_1[s]$ as infected, and the last step {\bf DD3} will be void.
\end{proof}

\subsection{Proof of \Lem~\ref{lemma_v0+}}\label{Sec_lemma_v0+}
Let $\cE$ be the event that properties (i) and (iii) from \Prop~\ref{Prop_basic} hold; then $\pr\brk\cE=1-o(n^{-2})$.
Moreover, let $\fE$ be the $\sigma$-algebra generated by $\SIGMA$ and the neighbourhoods $(\partial x)_{x\in V_1}$.
Then the event $\cE$ is $\fE$-measurable while the neighbourhoods $(\partial x)_{x\in V_0}$ of the healthy individuals are independent of $\fE$.
Recalling from {\bf SC1} that the individuals $x\in V_0[i]$ choose $\Delta/s$ random tests in each of the compartments $F[i+j]$, $0\leq j\leq s-1$ independently and remembering that $x\in V_{0+}[i]$ iff none of these tests is negative, on $\cE$ we obtain
\begin{align}\nonumber
\pr\brk{x\in V_{0+}[i]\mid \fE}&=\binom{m/(2\ell)+O(\sqrt m\log^3n)}{\Delta/s}^s\binom{m/\ell}{\Delta/s}^{-s}
=\bcfr{1+O(m^{-1/2}\ell\log^3n)}{2}^\Delta\\
&=2^{-\Delta}+O(m^{-1/2}\Delta\ell\log^3n)=2^{-\Delta}(1+O(n^{-\theta/2}\log^4n))&&\mbox{[due to \eqref{eqell} and \eqref{eqDelta}].}
\label{eqlemma_v0+1}
\end{align}
Because all $x\in V_0[i]$ choose their neighbourhoods independently, \eqref{eqlemma_v0+1} implies that the conditional random variable $|V_{0+}[i]|$ given $\fE$ has distribution $\Bin(|V_0[i]|,2^{-\Delta}(1+O(n^{-\Omega(1)})))$.
Therefore, since on $\cE$ we have $|V_0[i]|=|V[i]|+O(n^\theta)=n/\ell+O(n^\theta)$, the assertion follows from the Chernoff bound from \Lem~\ref{lem_chernoff}.

\subsection{Proof of \Lem~\ref{lem_1dev}}\label{Sec_lem_1dev}
The aim is to estimate the weighted sum $\vW_x^\star$ for infected individuals $x\in V[i+1]$ with $s\leq i<\ell$.
These individuals join tests in the $s$ compartments $F[i+j]$, $j\in[s]$.
Conversely, for each such $j$ the tests $a\in F[i+j]$ recruit their individuals from the compartments $V[i+j-s+1],\ldots,V[i+j]$.
Thus, the compartments preceding $V[i+1]$ that the tests in $F[i+j]$ draw upon are $V[h]$ with $i+j-s<h\leq i$.
We begin by investigating the set $\cW_{i,j}$ of tests $a\in F[i+j]$ without an infected individual from these compartments, i.e.,
\begin{align*}
\cW_{i,j}={\cbc{a\in F[i+j]:(V_1[1]\cup\cdots\cup V_1[i])\cap\partial a=\emptyset}}={\cbc{a\in F[i+j]:\bigcup_{i+j-s+1<h\leq i}V_1[h]\cap\partial a=\emptyset}}.
\end{align*}

\begin{claim}\label{Claim_Qxj}
With probability $1-o(n^{-2})$ for all $s\leq i<\ell$, $j\in[s]$ we have $|\cW_{i,j}|=2^{-(s-j)/s}\frac m\ell(1+O(n^{-\Omega(1)}))$.
\end{claim}
\begin{proof}
We may condition on the event $\cE$ that (i) from \Prop~\ref{Prop_basic} occurs.
To compute the mean of $|\cW_{i,j}|$ fix a test $a\in F[i+j]$ and an index $i+j-s<h\leq i$.
Then \eqref{eqax} shows that the probability that a fixed individual $x\in V[h]$ joins $a$ equals
$\pr\brk{x\in\partial a}=\frac{\Delta\ell}{ms}(1+O(\frac{\Delta\ell}{ms}))$.
Hence, the choices \eqref{eqell} and \eqref{eqDelta} of $\Delta$ and $\ell$  and the assumption $m=\Theta(k\log n)$ ensure that
\begin{align}\nonumber
\Erw\brk{\abs{(V_1[i+j-s+1]\cup\cdots\cup V_1[i])\cap\partial a}\mid\cE}&=\bc{s-j}\bc{\frac{\Delta\ell}{ms}\cdot\frac{k}{\ell}+
	O\bcfr{\Delta^2k}{m^2s^2}+O\bcfr{\Delta\ell\sqrt k\log n}{ms}}\\
	&=\frac{s-j}s\log 2+O(n^{-\Omega(1)}).\label{eqClaim_Amin1_1}
\end{align}
Since by {\bf SC1} the events $\{x\in\partial a\}_x$ are independent,  $|V_1[h]\cap\partial a|$ is a binomial random variable for every $h$ and all these random variables $(|V_1[h]\cap\partial a|)_h$ are mutually independent.
Therefore, \eqref{eqClaim_Amin1_1} implies that
\begin{align}\label{eqClaim_Amin1_1a}
\pr\brk{(V_1[i+j-s+1]\cup\cdots V_1[i])\cap\partial a=\emptyset\mid\cE}=2^{-(s-j)/s}+O(n^{-\Omega(1)}).
\end{align}
Hence,
\begin{align}\label{eqClaim_Amin1_2}
\Erw\brk{|\cW_{i,j}|\mid\cE}&=\sum_{a\in F[i+j]}
\pr\brk{(V_1[i+j-s+1]\cup\cdots\cup V_1[i])\cap\partial a=\emptyset\mid\cE}=\frac m\ell2^{-(s-j)/s}(1+O(n^{-\Omega(1)})).
\end{align}
Finally, changing the neighbourhood $\partial x$ of one infected individual $x\in V_1$ can alter $|\cW_{i,j}|$ by at most $\Delta$.
Therefore, the Azuma--Hoeffing inequality shows that for any $t>0$,
\begin{align}\label{eqClaim_Amin1_3}
\pr\brk{\abs{|\cW_{i,j}|-\Erw[|\cW_{i,j}|\mid\cE]}>t\mid\cE}&\leq2\exp\bc{-\frac{t^2}{2k\Delta^2}}.
\end{align}
Combining \eqref{eqClaim_Amin1_2} and \eqref{eqClaim_Amin1_3}, applied with $t=\sqrt m\log^2 n$, and taking a union bound on $i,j$ completes the proof.
\end{proof}

As a next step we use Claim~\ref{Claim_Qxj} to estimate the as yet unexplained tests counts $\vW_{x,j}$ from \eqref{eqWxj}.

\begin{claim} \label{Claim_unexp}
For all $s \leq i < \ell$, $x \in V_1[i+1]$ and $j \in [s]$ we have
\begin{align*}
\pr\brk{\wxj xj<(1-\eps/2)2^{j/s-1}\Delta/s}\leq\exp\bc{-\frac{\Omega(\log n)}{(\log\log n)^4}}.
\end{align*}
\end{claim}			
\begin{proof}
Fix a pair of indices $i,j$ and an individual $x\in V_1[i+1]$.
We also condition on the event $\cE$ that (i) from \Prop~\ref{Prop_basic} occurs.
Additionally, thanks to Claim~\ref{Claim_Qxj} we may condition on the event
\begin{align*}
\cE'&=\cbc{|\cW_{i,j}|=2^{-(s-j)/s}\frac m\ell(1+O(n^{-\Omega(1)}))}.
\end{align*}
Further, let $\fE$ be the $\sigma$-algebra generated by $\SIGMA$ and by the neighbourhoods $(\partial y)_{y\in V[1]\cup\cdots\cup V[i]}$.
Recall from {\bf SC1} that $x$ simply joins $\Delta/s$ random tests in compartment $F[i+j]$, independently of all other individuals, and remember from \eqref{eqWxj} that $\wxj xj$ counts tests $a\in \cW_{i+j}\cap\partial x$.
Therefore, since the events $\cE,\cE'$ and the random variable $|\cW_{i,j}|$ are $\fE$-measurable while $\partial x$ is independent of $\fE$,  given $\fE$ the random variable $\wxj xj$ has a hypergeometric distribution $\Hyp(m/\ell,|\cW_{i,j}|,\Delta/s)$.
Thus, the assertion follows from the hypergeometric Chernoff bound from \Lem~\ref{lem_hyperchernoff} and the choice \eqref{eqwj} of $\zeta$.
\end{proof}
			
\begin{proof}[Proof of \Lem~\ref{lem_1dev}]
Since $\wxstar x=\sum_{j=1}^s w_j \wxj xj$, the lemma is an immediate consequence of Markov's inequality and Claim~\ref{Claim_unexp}.
\end{proof}
			
\subsection{Proof of \Lem~\ref{lem_large_deviation0+}}\label{Sec_lem_large_deviation0+}
We need to derive the rate functions of the random variable $\wxj xj$ that count as yet unexplained tests for $x \in \zeroplus[i+1]$.
To this end we first investigate the set of positive tests in compartment $i+j$ that do not contain any infected individuals from the first $i$ compartments.
In symbols,
\begin{align*}
	\cP_{i+1,j}&=\cbc{a\in F_1[i+j]:\partial a\cap(V_1[1]\cup\cdots\cup V_1[i])=\emptyset}&&(s\leq i<\ell,\,j\in[s]).
\end{align*}

\begin{claim}\label{lem_U}
\Whp\ for all $s\leq i<\ell,j\in[s]$ we have $|\cP_{i+1,j}| = \bc{1+O(n^{-\Omega(1)})} \bc{2^{j/s}-1}\frac m{2\ell}$.
\end{claim}
\begin{proof}
We may condition on the event $\cE$ that (i) from \Prop~\ref{Prop_basic} occurs.
As a first step we calculate the probability that $(V_1[i+1]\cup\cdots\cup V_1[i+j])\cap\partial a\neq\emptyset$ for a specific test $a\in F[i+j]$.
To this end we follow the steps of the proof of Claim~\ref{Claim_Qxj}.
Since by \eqref{eqax} a specific individual $x\in V[h]$, $i<h\leq i+j$, joins $a$ with probability $\pr\brk{x\in\partial a}=(\Delta\ell/(ms))(1+O(\Delta\ell/(ms)))$ and since given $\cE$ each compartment $V[h]$ contains $k/\ell+O(\sqrt{k/\ell}\log n)$ infected individuals, we obtain, in perfect analogy to \eqref{eqClaim_Amin1_1},
\begin{align}\label{eqClaim_Amin2_1}
\Erw\brk{\abs{(V_1[i+1]\cup\cdots\cup V_1[i+j])\cap\partial a}\mid\cE}&=\frac{j}{s}\log 2+O(n^{-\Omega(1)}).
\end{align}
Since the individuals $x\in V[i+1]\cup\cdots\cup V[i+j]$ join tests independently, \eqref{eqClaim_Amin2_1} implies that
\begin{align}\label{eqClaim_Amin2_2}
\pr\brk{(V_1[i+1]\cup\cdots\cup V_1[i+j])\cap\partial a\neq\emptyset\mid\cE}&=1-2^{-j/s}+O(n^{-\Omega(1)}).
\end{align}
Furthermore, we already verified in \eqref{eqClaim_Amin1_1a} that
\begin{align}\label{eqClaim_Amin2_2a}
\pr\brk{(V_1[i+j-s+1]\cup\cdots V_1[i])\cap\partial a=\emptyset\mid\cE}=2^{-(s-j)/s}+O(n^{-\Omega(1)}).
\end{align}
Because the choices for the compartments $V[i+j-s+1]\cup\cdots\cup V[i+j]$ from which $a$ draws its individuals are mutually independent, we can combine \eqref{eqClaim_Amin2_2} with \eqref{eqClaim_Amin2_2a} to obtain
\begin{align}
\pr\brk{\bigcup_{i+j-s<h\leq i} V_1[h]\cap\partial a=\emptyset\neq\bigcup_{i<h\leq i+j}V_1[h]\cap\partial a\mid\cE}
&=\frac{2^{j/s}-1}2+O(n^{-\Omega(1)}).
\label{eqClaim_Amin2_2b}
\end{align}
Further, \eqref{eqClaim_Amin2_2b} implies
\begin{align}\label{eqClaim_Amin2_3}
\Erw\brk{|\cP_{i+1,j}|\mid\cE}=
\Erw\brk{\abs{\cbc{a\in F_1[i+j]:
\bigcup_{h\leq i} V_1[h]\cap\partial a=\emptyset\neq\bigcup_{i<h}V_1[h]\cap\partial a}}
\mid\cE}&=(2^{j/s}-1)\frac{m}{2\ell}\bc{1+O(n^{-\Omega(1)})}.
\end{align}
Finally, altering the neighbourhood $\partial x$ of any infected individual can shift $|\cP_{i+1,j}|$  by at most $\Delta$.
Therefore, the Azuma--Hoeffding inequality implies that
\begin{align}\label{eqClaim_Amin2_4}
\pr\brk{\abs{|\cP_{i+1,j}|-\Erw[|\cP_{i+1,j}|\mid\cE]}>t\mid\cE}&\leq2\exp\bc{-\frac{t^2}{2k\Delta^2}}.
\end{align}
Thus, the assertion follows from \eqref{eqDelta}, \eqref{eqClaim_Amin2_3} and \eqref{eqClaim_Amin2_4} by setting $t=\sqrt m\log^2 n$.
\end{proof}

Thanks to \Prop~\ref{Prop_basic} (iii) and \Lem~\ref{lem_U} in the following we may condition on the event
\begin{align}
\cU=\cbc{\forall s<i\leq \ell,j\in[s]:
|F_1[i+j]|=\bc{1 + O(n^{-\Omega(1)})}\frac m{2\ell}\wedge
|\cP_{i+1,j}| = \bc{1 + O(n^{-\Omega(1)})} \bc{2^{j/s}-1}\frac m{2\ell}}. \label{def_event_U}
\end{align}
As a next step we will determine the conditional distribution of $\wxj xj$ for $x\in V_{0+}[i+1]$ given $\cU$.
			
\begin{claim} \label{Lemma_wxj_xj}
Let $s<i\leq\ell$ and $j\in[s]$.
Given $\cU$ for every $x \in V_{0+}[i+1]$ we have
\begin{align}\label{eqLemma_wxj_xj}
\wxj xj \sim \Hyp \bc{\bc{1 + O(n^{-\Omega(1)})}\frac m{2\ell},\bc{1 + O(n^{-\Omega(1)})} \bc{2^{j/s}-1}\frac m{2\ell},\frac\Delta s}.
\end{align}
\end{claim}
\begin{proof}
By {\bf SC1} each individual $x \in V_{0+}[i+1]$ joins $\Delta/s$ positive test from $F[i+j]$, drawn uniformly without replacement.
Moreover, by \eqref{eqWxj} given $x \in V_{0+}[i+1]$ the random variable $\wxj xj$ counts the number of tests $a\in\cP_{i+1,j}\cap\partial x$.
Therefore, $\wxj xj\sim \Hyp(|F_1[i+j],|\cP_{i+1,j}|,\Delta/s)$.
Hence, given $\cU$ we obtain \eqref{eqLemma_wxj_xj}.
\end{proof}

\noindent
The estimate \eqref{eqLemma_wxj_xj} enables us to bound the probability that $\wxstar x$ gets `too large'.

\begin{claim}\label{lem_happen1}
Let
\begin{align*}
	\cM =&\min\frac1s\sum_{j=1}^{s-1} \vecone\cbc{z_j \geq 2^{j/s}-1} \KL{z_j}{2^{j/s}-1}\\
	 &\mbox{s.t.}\quad\sum_{j=1}^{s-1}\bc{z_j-(1-2\zeta)2^{j/s-1}} w_j = 0,\qquad z_1,\ldots,z_{s-1}\in[0,1].
\end{align*}
Then for all $s \leq i < \ell$ and all $x\in V[i+1]$ we have
\begin{align*}
	\pr\brk{\wxstar x > (1-2 \zeta)\frac{\Delta}{s} \wq \mid \cU,\,x\in \zeroplusi {i+1}} &\leq \exp(-(1+o(1))\cM \Delta).
\end{align*}
\end{claim}
\begin{proof} 
Let $s \leq i < \ell$ and $x\in \zeroplusi {i+1}$.
Step {\bf SC1} of the construction of $\G$ ensures that the random variables $(\vW_{x,j})_{j\in[s]}$ are independent because the tests in the various compartments $F[i+j]$, $j\in[s]$, that $x$ joins are drawn independently.
Therefore, the definition \eqref{eqWstar} of $\wxstar x$ and \Lem~\ref{Lemma_wxj_xj} yield
\begin{align}
\Pr&\brk{ \wxstar x > (1-2 \zeta)\frac{\Delta}{s} \wq \mid \cU,\,x\in \zeroplusi {i+1}}=\pr\brk{\sum_{j=1}^{s-1}w_j\vW_{x,j}\geq \frac{1-2 \zeta}{s} \wq \mid \cU,\,x\in \zeroplusi {i+1}}
 \nonumber\\
	& \leq \sum_{y_1, \ldots, y_s=0}^{\Delta}
		\vecone\cbc{\sum_{j=1}^{s-1} w_j y_j\geq\frac{1-2 \zeta}{s} \wq }
		\prod_{j=1}^{s-1} \Pr[\vW_{x,j}\geq y_j \mid \cU,\,x\in \zeroplusi {i+1}].
\label{Eq_Lagrange_1}
\end{align}
Further, let
\begin{align*}
\cZ&=\cbc{(z_1,\ldots,z_{s-1})\in[0,1]^{s-1}:\sum_{j=1}^{s-1}\bc{z_j-(1-2\zeta)2^{j/s-1}} w_j = 0}.
\end{align*}
Substituting  $y_j=\Delta z_j/s$ in \eqref{Eq_Lagrange_1} and bounding the total number of summands by $(\Delta+1)^s$, we obtain
\begin{align}\label{Eq_Lagrange_2}
\Pr\brk{ \wxstar x > (1-2 \zeta)\frac{\Delta}{s} \wq \mid \cU,\,x\in \zeroplusi {i+1}}&\leq(\Delta+1)^s\max_{(z_1,\ldots,z_s)\in\cZ}\prod_{j=1}^{s-1}\Pr[\vW_{x,j}\geq \Delta z_j/s \mid \cU,\,x\in \zeroplusi {i+1}].
\end{align}
Moreover, Claim~\ref{Lemma_wxj_xj} and the Chernoff bound from \Lem~\ref{lem_hyperchernoff} yield
\begin{align*}
\Pr[\vW_{x,j}\geq \Delta z_j/s \mid \cU,\,x\in \zeroplusi {i+1}]
\leq\exp\bc{-\vecone\cbc{z_j\geq p_j}\frac\Delta s\KL{z_j}{p_j}}&&\mbox{where }p_j=2^{j/s}-1+O(n^{-\Omega(1)}).
\end{align*}
Consequently, since \eqref{eqDelta} and the assumption $m=\Theta(k\log n)$ ensure that $\Delta=\Theta(\log n)$, we obtain
\begin{align}\label{Eq_Lagrange_3}
\Pr[\vW_{x,j}\geq \Delta z_j/s \mid \cU,\,x\in \zeroplusi {i+1}]
\leq\exp\bc{-\vecone\cbc{z_j\geq 2^{j/s}-1}\frac\Delta s\KL{z_j}{2^{j/s}-1}+O(n^{-\Omega(1)})}.
\end{align}
Finally, the assertion follows from \eqref{Eq_Lagrange_2} and \eqref{Eq_Lagrange_3}.
\end{proof}
			
\noindent
As a next step we solve the optimisation problem $\cM$ from Claim~\ref{lem_happen1}.

\begin{claim} \label{lem_lagrange}
We have $\cM=1-\log2+O(\log(s)/s)$.
\end{claim}
\begin{proof}
Fixing an auxiliary parameter $\delta\geq0$ we set up the Lagrangian
\begin{align*}
	\cL_\delta(z_1,\ldots,z_s,\lambda)&=\sum_{j=1}^{s-1}
		\bc{\vecone\cbc{z_j\geq2^{j/s}-1}+\delta\vecone\cbc{z_j<2^{j/s}-1}}
		\KL {z_j}{2^{j/s}-1}+\frac\lambda s\sum_{j=1}^{s-1}w_j \bc{z_j - (1 - 2 \zeta) 2^{j/s-1}}.
\end{align*}
The partial derivatives come out as
\begin{align*}
	\frac{\partial\cL_\delta}{\partial\lambda}&=-\frac1s\sum_{j=1}^{s-1} ((1 - 2 \zeta) 2^{j/s-1}-z_j) w_j,&
	\frac{\partial\cL_\delta}{\partial z_j}&=-\lambda w_j+\bc{\vecone\cbc{z_j\geq2^{j/s}-1}+\delta\vecone\cbc{z_j<2^{j/s}-1}}\log\frac{z_j(2-2^{j/s})}{(1-z_j)(2^{j/s}-1)}.
\end{align*}
Set $z_j^*=(1 - 2 \zeta) 2^{j/s-1}$ and $\lambda^*=1$.
Then clearly
\begin{align}\label{eqLdiff1}
\frac{\partial\cL_\delta}{\partial\lambda}\bigg|_{\lambda^*,z_1^*,\ldots,z_{s-1}^*}=0.
\end{align}
Moreover, the choice \eqref{eqwj} of $\zeta$ guarantees that $z_j^*\geq2^{j/s}-1$.
Hence, by the choice \eqref{eqwj} of the weights $w_j$,
\begin{align}\label{eqLdiff2}
\frac{\partial\cL_\delta}{\partial z_j}\bigg|_{\lambda^*,z_1^*,\ldots,z_{s-1}^*}=0.
\end{align}
Since $\cL_\delta(y_1,\ldots,y_s,\lambda)$ is strictly convex in $z_1,\ldots,z_s$ for every $\delta>0$, \eqref{eqLdiff1}--\eqref{eqLdiff2} imply that $\lambda^*,z_1^*,\ldots,z_{s-1}^*$ is a global minimiser.
Furthermore, since this is true for any $\delta>0$ and since $z_j^*\geq2^{j/s}-1$, we conclude that $(z_1^*,\ldots,z_{s-1}^*)$ is an optimal solution to the minimisation problem $\cM$.
Hence,
\begin{align}\label{eqlem_lagrange10}
\cM= \frac 1s \sum_{j=1}^{s-1} \KL {z_j^*}{2^{j/s}-1}=\frac1s\sum_{j=1}^{s-1}\KL {(1 - 2 \zeta) 2^{j/s-1}}{2^{j/s}-1}. 
\end{align}
Since
\begin{align*}
\frac\partial{\partial\alpha}\KL {(1 - 2 \alpha) 2^{z-1}}{2^{z}-1}&=2^z\brk{-z\log(2)+\log(1-2^{z-1}+\alpha2^z)-\log(1-2^{z-1})-\log(1-2\alpha)+\log(2^z-1)},
\end{align*}
we obtain $\frac\partial{\partial\alpha}\KL {(1 - 2 \alpha) 2^{z-1}}{2^{z}-1}=O(\log s)$ for all $z=1/s,\ldots,(s-1)/s$ and $\alpha\in[0,2\zeta]$.
Combining this bound with \eqref{eqlem_lagrange10}, we arrive at the estimate
\begin{align}\label{eqlem_lagrange11}
\cM= O(\zeta\log s)+\frac1s\sum_{j=1}^{s-1}\KL {2^{j/s-1}}{2^{j/s}-1}. 
\end{align}
Additionally, the function $f: z\in[0,1]\mapsto\KL{2^{z-1}}{2^z-1}$ is strictly decreasing and convex.
Indeed,
\begin{align*}
   f'(z) &=  \frac{2^{z-1}\log 2}{2^z-1} \bc{(2^z-1) \log \bc{\frac{2^z}{2^z-1}}-1},&
   f''(z) &= \bc{2^{z-1}\log^22}\bc{\log \bc{\frac{2^z}{2^z-1}} +\frac{2 - 2^z}{(2^z-1)^2} }.
\end{align*}
The first derivative is negative because $2^{z-1}/(2^z-1)>0$ while $(2^z-1) \log \bc{2^z/(2^z-1)}<1$ for all $z \in (0,1)$.
Moreover, since evidently $f''(z)>0$ for all $z\in(0,1)$, we obtain convexity.
Further, l'H\^opital's rule yields $$\KL{2^{1/s-1}}{2^{1/s}-1}=O(\log s).$$
As a consequence, we can approximate the sum \eqref{eqlem_lagrange11} by an integral and obtain
\begin{align*}
\cM&=O(\log(s)/s)+\int_0^1 \KL {2^{z-1}}{2^{z}-1} \dd z \\
&= O(\log(s)/s)+ \frac{2(1-z) \log^2(2) + 2^z \log 2^z + (1-2^z) \log (2^z-1)}{2\log 2}\bigg|_{z=0}^{z=1} = 1-\log(2) + O(\log(s)/s),
\end{align*}
as claimed.
\end{proof}

\begin{proof}[Proof of \Lem~\ref{lem_large_deviation0+}]
Fix $s\leq i<\ell$ and let $\vX_i$ be the number of  $x\in\zeroplus[i]$ such that  $\wxstar x > (1-2\zeta) \frac \Delta s \wq$.
Also recall that \Prop~\ref{Prop_basic} (iii) and Claim~\ref{lem_U} imply that $\pr\brk\cU=1-o(1)$.
Combining \Lem~\ref{lemma_v0+} with Claims~\ref{lem_happen1} and~\ref{lem_lagrange}, we conclude that
\begin{align}\label{lem_large_deviation0+_1}
\Erw[\vX_i\mid\cU]&\leq
\bc{1+O \bc{n^{-\Omega(1)}}} 2^{-\Delta} n\exp(-(1-\log(2)+o(1)) \Delta)=\exp\bc{\log n-(1+o(1))\Delta}.
\end{align}
Recalling the definition \eqref{eqDelta} of $\Delta$ and using the assumption that $m\geq(1+\eps)\madapt$ for a fixed $\eps>0$, we obtain  $\Delta\geq(1-\theta+\Omega(1))\log n$.
Combining this estimate with \eqref{lem_large_deviation0+_1}, we find
\begin{align}\label{lem_large_deviation0+_3}
\Erw[\vX_i\mid\cU]&\leq n^{\theta-\Omega(1)}.
\end{align}
Finally, the assertion follows from \eqref{lem_large_deviation0+_3} and Markov's inequality.
\end{proof}

\subsection{Proof of \Prop~\ref{prop_dist_psi}}\label{sec_prop_dist_psi}

The following lemma establishes an expansion property of $\G$.
Specifically, if $T$ is a small set of individuals, then there are few individuals $x$ that share many tests with another individual from $T$.

\begin{lemma} \label{lemma_endgame_misclassified}
Suppose that $m=\Theta(n^\theta\log n)$.
\Whp\ for any set $T \subset V$ of size at most $\exp(-\log^{7/8} n)k$ we have
\begin{align*}
\abs{\cbc{x\in V:\sum_{a\in\partial x\setminus F[0]}\vecone\cbc{T\cap\partial a\setminus \cbc x\neq\emptyset}\geq\ln^{1/4}n}}\leq\frac{|T|}3.
\end{align*}
\end{lemma}
\begin{proof}
Fix a set $T\subset V$ of size $t=|T|\leq \exp(-\log^{7/8} n)k$, a set $R\subset V$ of size $r=\lceil t/3\rceil $ and let $\gamma=\lceil\ln^{1/4}n\rceil$.
Furthermore, let $U\subset F[1]\cup\cdots\cup F[\ell]$ be a set of tests of size $\gamma r\leq u\leq \Delta t$.
Additionally, let $\cE(R,T,U)$ be the event that every test $a\in U$ contains two individuals from $R\cup T$.
Then
\begin{align}\label{eq_lemma_endgame_misclassified_1}
\pr\brk{R\subset\cbc{x\in V:\sum_{a\in\partial x\setminus F[0]}\vecone\cbc{T\cap\partial a\setminus \cbc x\neq\emptyset}\geq\gamma}}
\leq\pr\brk{\cE(R,T,U)}.
\end{align}
Hence, it suffices to estimate $\pr\brk{\cE(R,T,U)}$.

Given a test $a\in U$ there are at most $\binom{r+t}2$ way to choose two individuals $x_a,x_a'\in R\cup T$.
Moreover, \eqref{eqax} shows that the probability of the event $\{x_a,x_a'\in\partial a\}$ is bounded by $(1+o(1))(\Delta\ell/(ms))^2$.
Therefore,
\begin{align*}
\pr\brk{\cE(R,T,U)}&\leq \brk{\binom{r+t}2\bcfr{(1+o(1))\Delta\ell}{ms}^2}^u.
\end{align*}
Consequently, the event $\cE(t,u)$ that there exist sets $R,T,U$ of sizes $|R|=r=\lceil t/3\rceil,|T|=t,|U|=u$ such that $\cE(R,T,U)$ occurs has probability
\begin{align*}
\pr\brk{\cE(t,u)}&\leq \binom nr\binom nt\binom mu\brk{\binom{r+t}2\bcfr{(1+o(1))\Delta\ell}{ms}^2}^u.
\end{align*}
Hence, the bounds $\gamma t/3\leq\gamma r\leq u\leq \Delta t$ yield
\begin{align*}
\pr\brk{\cE(t,u)}&\leq \binom nt^2\binom m{u}\brk{\binom{2t}2\bcfr{(1+o(1))\Delta\ell}{ms}^2}^{u}
	\leq\bcfr{\eul n}{t}^{2t}\bcfr{2\eul \Delta^2\ell^2t^2}{ms^2u}^u\\
	&\leq\brk{\bcfr{\eul n}{t}^{3/\gamma}\frac{6\eul \Delta^2\ell^2t}{\gamma ms^2}}^u
	\leq\brk{\bcfr{\eul n}{t}^{3/\gamma}\cdot\frac{t\ln^4n}m}^u&&\mbox{[due to \eqref{eqell}, \eqref{eqDelta}]}.
\end{align*}
Further, since $\gamma=\Omega(\log^{1/4}n)$ and $m=\Omega(k\log n)$ while $t\leq \exp(-\log^{7/8} n)k$, we obtain $\pr\brk{\cE(t,u)}\leq \exp(-u\sqrt{\log n})$.
Thus,
\begin{align}\label{eq_lemma_endgame_misclassified_2}
\sum_{\substack{1\leq t\le k^{1-\alpha}\\\gamma t/3\leq u\leq \Delta t}}\pr\brk{\cE(t,u)}
&\leq\sum_{1\leq u\leq \Delta t}u\exp(-u\sqrt{\log n})=o(1).
\end{align}
Finally, the assertion follows from \eqref{eq_lemma_endgame_misclassified_1} and \eqref{eq_lemma_endgame_misclassified_2}.
\end{proof}

\begin{proof}[Proof of \Prop~\ref{prop_dist_psi}]
With $\tau$ the result of steps 1--10 of \SPIV\ let $\cM[i] = \cbc{ x \in V[i]: \tau_x \neq \SIGMA_x  }$ be the set of misclassified individuals in compartment $V[i]$.
\Prop~\ref{prop_seed} shows that \whp\ $\cM[i]=\emptyset$ for all $i\leq s$.
Further, we claim that for every $s\leq i<\ell$ and any individual $x\in\cM[i+1]$ one of the following three statements is true.
\begin{description}
    \item[M1] $x\in V_1[i+1]$ and $\wxstar x < (1- \zeta/2) \frac \Delta s \wq$,
    \item[M2] $x \in \zeroplusi{i+1}$ and $ \wxstar x > (1-2\zeta) \frac \Delta s \wq$, or
    \item[M3] $x \in V[i+1]$ and $\sum_{a\in\partial x}\vecone\{\partial a\cap(\cM[1]\cup\cdots\cup\cM[i])\neq\emptyset\}\geq\ln^{1/4} n$.
\end{description}
To see this, assume that $x\in\cM[i+1]$ while {\bf M3} does not hold.
Then comparing \eqref{eqWxj} and \eqref{eqWtau} we obtain
\begin{align}\label{eqprop_dist_psi_9}
\abs{W_{x,j}(\tau)-\vW_{x,j}}&\leq\ln^{1/4}n&&\mbox{for all }1\leq j<s.
\end{align}
Moreover, the definition \eqref{eqwj} of the weights, the choice \eqref{eqs} of $s$, and the choices \eqref{eqwj} of $\zeta$ and the weights $w_j$ ensure that $0\leq w_j\leq O(s)=O(\ln\ln n)$.
This bound implies together with the definition \eqref{eqWstar} of the scores $\wxstar x$ and \eqref{eqprop_dist_psi_9} that
\begin{align}\label{eqprop_dist_psi_10}
|\wxstar{x}-W_x^\star(\tau)|=o(\zeta\Delta).
\end{align}
Thus, combining \eqref{eqprop_dist_psi_10} with the definition of $\tau_x$ in Steps 5--10 of \SPIV, we conclude that either {\bf M1} or {\bf M2} occurs.

Finally, to bound $\cM[i+1]$ let $\cM_1[i+1]$, $\cM_2[i+1]$, $\cM_3[i+1]$ be the sets of individuals $x\in V[i+1]$ for which {\bf M1}, {\bf M2} or {\bf M3} occurs, respectively.
Then \Lem s \ref{lem_1dev} and \ref{lem_large_deviation0+} imply that \whp{}
\begin{align*}
\abs{\cM_1[i+1]},\abs{\cM_2[i+1]}\leq k\exp\bc{-\frac{\ln n}{(\ln\ln n)^5}}.
\end{align*}
Furthermore, \Lem~\ref{lemma_endgame_misclassified} shows that  $\abs{\cM_3[i+1]}\leq \sum_{h=1}^i\abs{\cM[h]}$ \whp\
Hence, we obtain the relation
\begin{align}\label{eqprop_dist_psi_13}
\abs{\cM[i+1]}\leq k\exp\bc{-\frac{\ln n}{(\ln\ln n)^5}}+\sum_{h=1}^i\abs{\cM[h]}.
\end{align}
Because \eqref{eqell} ensures that the total number of compartments is $\ell=O(\ln^{1/2}n)$, the bound \eqref{eqprop_dist_psi_13} implies that $\abs{\cM[i+1]}\leq O(\ell^2k\exp(-(\ln n)/(\ln\ln n)^{5})$ for all $i\in[\ell]$ \whp{}
Summing on $i$ completes the proof.
\end{proof}


\subsection{Proof of \Prop~\ref{prop_endgame}}\label{sec_prop_endgame}

\noindent
For an infected individual $x\in V$ let
\begin{align*}
\vS_x[j]&= \abs{ \cbc{ a \in  F[j]\cap\partial x: V_1\cap\partial a=\cbc x}}&\mbox{and}&&\vS_x&=\sum_{j=1}^\ell\vS_x[j].
\end{align*}
Thus, $\vS_x[j]$ is the number of positive sets $a\in F[j]$ that $x$ has to itself, i.e., tests that do not contain a second infected individual, and $\vS_x$ is the total number of such tests.

\begin{lemma} \label{Lem_distphixstar}
Assume that $m\geq(1+\eps)\minf$.
\Whp\ we have $\min_{x\in\one}\vS_x\geq\sqrt \Delta$.
\end{lemma}
\begin{proof}
Due to \Prop~\ref{Prop_basic} we may condition on the event
\begin{align*}
    \cN = \cbc{\forall i\in[\ell]:\frac{m}{2\ell}-\sqrt m\ln n\leq\abs{F_0[i]}\leq\frac m{2\ell}+\sqrt m\ln n}.
\end{align*}
We claim that given $\cN$ for each $x \in V_{1}[i]$, $i\in[\ell]$, the random variable $\vS_x$ has distribution
\begin{align}\label{eqLem_distphixstar1}
\vS_x[i+j-1]\sim \Hyp\bc{\frac m\ell,\frac m{2\ell}+O(\sqrt m\log n),\frac\Delta s}.
\end{align}
To see this, consider the set $F_{x}[i+j-1] = \cbc{a \in F[i+j-1]: \partial a\cap \one \setminus\cbc x = \emptyset}$ of all tests in compartment $F[i+j-1]$ without an infected individual besides possibly $x$.
Since $x$ joins $\Delta/s=O(\ln n)$ tests in $F[i+j-1]$, given $\cN$ we have
\begin{align}\label{eqLem_distphixstar2}
\abs{F_{0,x}[i+j]}&=|F_0[i+j]|+O(\ln n)=\frac{m}{2\ell}+O(\sqrt m\ln n).
\end{align}
Furthermore, consider the experiment of first constructing the test design $\G$ and then re-sampling the set $\partial x$ of neighbours of $x$; i.e., independently of $\G$ we have $x$ join $\Delta/s$ random tests in each compartment $F[i+j]$.
Then the resulting test design $\G'$ has the same distribution as $\G$ and hence the random variable $\vS_x'[i+j-1]$ that counts tests $a\in F[i+j-1]\cap\partial x$ that do not contain another infected individual has the same distribution as $\vS_x[i+j-1]$.
Moreover, the conditional distribution of $\vS_x'[i+j-1]$ given $\G$ reads
\begin{align}\label{eqeqLem_distphixstar3}
\vS_x'[i+j-1]\sim\Hyp\bc{\frac m\ell,|F_{0,x}[i+j-1]|,\frac\Delta s}.
\end{align}
Combining \eqref{eqLem_distphixstar2} and \eqref{eqeqLem_distphixstar3}, we obtain \eqref{eqLem_distphixstar1}.

To complete the proof we combine \eqref{eqLem_distphixstar1} with \Lem~\ref{lem_hyperchernoff}, which implies that
\begin{align}
\pr\brk{\vS_x[i+j-1]\leq \sqrt\Delta\mid x\in V_1}&\leq\exp\bc{-\frac \Delta s\KL{(1+o(1))s/\sqrt{\Delta}}{1/2+o(1)}}
=\exp\bc{-(1+o(1))\frac{\Delta\log 2}{s}}.\label{eqeqLem_distphixstar4}
\end{align}
Since {\bf SC1} ensures that the random variables $(\vS_x[i+j-1])_{j\in[s]}$ are mutually independent, \eqref{eqeqLem_distphixstar4} yields
\begin{align}\label{eqeqLem_distphixstar5}
\pr\brk{\vS_x\leq \sqrt\Delta\mid x\in V_1}&\leq2^{-(1+o(1))\Delta}.
\end{align}
Finally, the assumption $m\geq(1+\eps)\minf$ for a fixed $\eps>0$ and the choice \eqref{eqDelta} of $\Delta$ ensure that $2^{-(1+o(1))\Delta}=o(1/k)$.
Thus, the assertion follows from \eqref{eqeqLem_distphixstar5} by taking a union bound on $x\in V_1$.
\end{proof}
			 
\begin{proof}[Proof of \Prop~ \ref{prop_endgame}]
For $j = 1 \dots \ceil{\ln n}$, let
\begin{align*}
\cM_j = \cbc{ x \in V: \tau^{(j)}_x \neq \SIGMA_x}
\end{align*}
contain all individuals that remain misclassified at the $j$-th iteration of the clean-up step.
\Prop~\ref{prop_dist_psi} shows that \whp
\begin{align}\label{eqprop_endgame1}
\abs{\cM_1} \leq k\exp\bc{-\frac{\log n}{(\log\log n)^6}}.
\end{align}
Furthermore, in light of \Lem~\ref{Lem_distphixstar} we may condition on the event $\cA=\{\min_{x\in\one}\vS_x\geq\sqrt \Delta\}$.

We now claim that given $\cA$ for every $j\geq1$
\begin{align}\label{eqprop_endgame2}
\cM_{j+1}\subset{\cbc{x\in V: \sum_{a\in\partial x\setminus F[0]}\abs{\partial a \cap \cM_j\setminus\cbc x} \geq \ceil{\ln^{1/4} n}}}.
\end{align}
To see this, suppose that $x\in\cM_{j+1}$ and recall that the assumption $m\geq\minf$ and \eqref{eqDelta} ensure that $\Delta=\Omega(\ln n)$.
Also recall that \SPIV's Step~15 thresholds the number 
$$S_x(\tau^{(j)})=\sum_{a\in\partial x:\hat\SIGMA_a=1}\vecone\cbc{\forall y\in\partial a\setminus\cbc x:\tau^{(j)}_y=0}$$
of positive tests containing $x$ whose other individuals are deemed uninfected.
There are two cases to consider.
\begin{description}
\item[Case 1:  $x\in V_0$]
in this case every positive tests $a\in\partial x$ contains an individual that is actually infected.
Hence, if $\tau^{(j)}_y=0$ for all $y\in\partial a\setminus\cbc x$, then $\partial a\cap\cM_j\setminus\cbc x\neq\emptyset$.
Consequently, since Step~15 of \SPIV\ applies the threshold of $S_x(\tau^{(j)})\geq\log^{1/4}n$, there are at least $\ln^{1/4}n$ tests $a\in\partial x$ such that $\partial a\cap\cM_j\setminus\cbc x\neq\emptyset$.
\item[Case 2:  $x\in V_1$]
given $\cA$ every infected $x$ participates in at least $\vS_x\geq\sqrt\Delta=\Omega(\ln^{1/2}n)$ tests that do not actually contain another infected individual.
Hence, if $S_x(\tau^{(j)})\leq\log^{1/4}n$, then at least $\sqrt\Delta-\log^{1/4}n\geq\log^{1/4}n$ tests $a\in\partial x$ contain an individual from   $\cM_j\setminus\cbc x$.
\end{description}
Thus, we obtain \eqref{eqprop_endgame2}.
Finally, \eqref{eqprop_endgame1}, \eqref{eqprop_endgame2} and \Lem~\ref{lemma_endgame_misclassified} show that \whp\ $|\cM_{j+1}|\leq|\cM_{j}|/3$ for all $j\geq1$.
Consequently, $\cM_{\lceil\ln n\rceil}=\emptyset$ \whp\ 
\end{proof}

\section{Optimal adaptive group testing}\label{sec_prop_adaptive}


\noindent
In this final section we show how the test design $\G$ from \Sec~\ref{Sec_alg} can be extended into an optimal two-stage adaptive design.
The key observation is that \Prop~\ref{prop_dist_psi}, which summarises the analysis of the first two phases of \SPIV\ (i.e., steps 1--10) only requires $m\geq(1+\eps)\madapt$ tests.
In other words, the excess number $(1+\eps)(\minf-\madapt)$ of tests required for non-adaptive group testing is necessary only to facilitate the clean-up step, namely phase~3 of \SPIV.

Replacing phase~3 of \SPIV\ by a second test stage,  we obtain an optimal adaptive test design.
To this end we follow Scarlett~\cite{Scarlett_2018}, who observed that a single-stage group testing scheme that correctly diagnoses all but $o(k)$ individuals with $(1+o(1))\madapt$ tests could be turned into a two-stage design  that diagnoses all individuals correctly \whp{} with $(1+o(1))\madapt$ tests in total.
(Of course, at the time no such optimal single-stage test design and algorithm were known.)
The second test stage works as follows.
Let $\tau$ denote the outcome of phases~1 and~2 of \SPIV\ applied to $\G$ with $m=(1+\eps)\madapt$.
\begin{description}
\item[T1] Test every individual from the set $V_1(\tau)=\{x\in V:\tau_x=1\}$ of individuals that \SPIV\ diagnosed as infected separately.
\item[T2] To the individuals $V_0(\tau)=\{x\in V:\tau_x=0\}$ apply the random $d$-out design and the {\tt DD}-algorithm from \Sec~\ref
{Sec_random_bip} with a total of $m=k$ tests and $d=\lceil 10\log n\rceil $.
\end{description}
Let $\tau'\in\{0,1\}^V$ be the result of {\bf T1}--{\bf T2}.

\begin{proposition}\label{Prop_ad}
\Whp\ we have $\tau'_x=\SIGMA_x$ for all $x\in V$.
\end{proposition}

As a matter of course {\bf T1} renders correct results, i.e., for all individuals $x\in V_1(\tau)$ we have $\tau_x'=\SIGMA_x$.
Further, to analyse {\bf T2} we use a similar argument as in the analysis of the first phase of \SPIV\ in \Sec~\ref{Sec_prop_seed}; we include the analysis for the sake of completeness.
We begin by investigating the number of negative tests.
Let $\G'$ denote the test design set up by {\bf T2}, let $F'=\{b_1,\ldots,b_k\}$ denote its set of tests and let $\hat\SIGMA_{b_1},\ldots,\hat\SIGMA_{b_k}$ signify the corresponding test results.
Further, let $F'_0=\{b\in F':\hat\SIGMA_b=0\}$ and $F'_1=\{b\in F':\hat\SIGMA_b=1\}$ be the set of negative and positive tests, respectively.

\begin{lemma}\label{Lemma_ad}
\Whp\ we have $|F'_1|\leq\frac k2$.
\end{lemma}
\begin{proof}
\Prop~\ref{prop_dist_psi} implies that \whp{}
\begin{align}\label{eqLemma_ad_2}
|V_0(\tau)\cap V_1|\leq \sum_{x\in V}\vecone\cbc{\tau_x\neq\SIGMA_x}\leq k\exp\bc{-\frac{\log n}{(\ln\ln n)^6}}.
\end{align}
Moreover, since every individual $x\in V_0(\tau)$ joins $d$ random tests, for any specific test $b\in F'$ we have
\begin{align*}
\pr\brk{x\in\partial_{\G'}b}&=1-\pr\brk{x\not\in\partial_{\G'}b}=1-\binom{k-1}{d}\binom{k}{d}^{-1}=\frac dk(1+O(n^{-\Omega(1)})).
\end{align*}
Hence, for every test $b\in F'$,
\begin{align*}
\Erw\brk{|\partial b\cap V_1|\,\bigg|\, |V_0(\tau)\cap V_1|\leq k\exp\bc{-\frac{\log n}{(\ln\ln n)^6}}}=O(1/\log n).
\end{align*}
Consequently,
\begin{align}\label{eqLemma_ad_1}
\Erw\brk{|F_1'|\mid |V_0(\tau)\cap V_1|\leq k/\log n}=O(k/\log n).
\end{align}
Finally, combining \eqref{eqLemma_ad_2} and \eqref{eqLemma_ad_1} and applying Markov's inequality, we conclude that $|F'_1|\leq\frac k2$ \whp
\end{proof}

\begin{corollary}\label{Cor_ad}
\Whp\ for every $x\in V_0(\tau)$ there is a test $b\in F'$ such that $\partial b\setminus\cbc x\subset V_0$.
\end{corollary}
\begin{proof}
We construct the random graph $\G'$ in two rounds.
In the first round we first expose the neighbourhoods $(\partial_{\G'} y)_{y\in V_0(\tau)\setminus\cbc x}$.
\Lem~\ref{Lemma_ad} implies that after the first round the number $\vX$ of tests  that do not contain an infected individual $y\in V_0(\tau)\cap V_1$ exceeds $k/2$ \whp{}
In the second round we expose $\partial_{\G'}x$.
Because $\partial_{\G'}x$ is chosen independently of the neighbourhoods $(\partial_{\G'} y)_{y\in V_0(\tau)\setminus\cbc x}$, the number of tests $b\in\partial_{\G'}x$ that do not contain an infected individual $y\in V_0(\tau)\cap V_1$ has distribution $\Hyp(k,\vX,d)$.
Therefore, since $d\geq10\log n$ we obtain
\begin{align}\label{eqCor_ad}
\pr\brk{\forall b\in\partial x:V_1\cap\partial b\setminus\cbc x\neq\emptyset\mid\vX\leq k/2}&\leq\pr\brk{\Hyp(k,k/2,d)=0}
\leq2^{-d}=o(1/n).
\end{align}
Finally, the assertion follows \eqref{eqCor_ad} and the union bound.
\end{proof}

\begin{proof}[Proof of \Prop~\ref{Prop_ad}]
\Cor~\ref{Cor_ad} shows that we may assume that for every $x\in V_0(\tau)$ there is a test $b_x\in F'$ with $\partial b_x\setminus\cbc x\subset V_0$.
As a consequence, upon executing the first step {\bf DD1} of the {\tt DD} algorithm, {\bf T2} will correctly diagnose all individuals $x\in V_0(\tau)\cap V_0$.
Therefore, if $x\in V_0(\tau)\cap V_1$, then {\bf DD2} will correctly identify $x$ as infected because all other individuals $y\in\partial b_x$ were already identified as healthy by {\bf DD1}.
Thus, $\tau_x'=\SIGMA_x$ for all $x\in V$.
\end{proof}

\begin{proof}[Proof of \Thm~\ref{Thm_ad}]
\Prop~\ref{Prop_ad} already establishes that the output of the two-stage adaptive test is correct \whp{}
Hence, to complete the proof we just observe that the total number of tests comes to $(1+\eps)\madapt$ for the first stage plus $|V_1(\tau)|+k$ for the second stage.
Furthermore, \Prop~\ref{prop_dist_psi} implies that \whp{}
\begin{align*}
|V_1(\tau)|\leq |V_1|+\sum_{x\in V}\vecone\cbc{\tau_x\neq\SIGMA_x}\leq k\bc{1+\exp\bc{-\frac{\log n}{(\ln\ln n)^6}}}=(1+o(1))k.
\end{align*}
Thus, the second stage conducts $O(k)=o(\madapt)$ tests.
\end{proof}

\subsection*{Acknowledgment}
We thank Arya Mazumdar for bringing the group testing problem to our attention.


\end{document}